\documentclass[11pt]{article}

\usepackage{graphicx}
\usepackage{dcolumn}
\usepackage{bm}

\usepackage{geometry}

\usepackage{mathrsfs}  

\usepackage{appendix}

\usepackage[utf8]{inputenc}
\usepackage[T1]{fontenc}
\usepackage{color}
\usepackage{physics}
\usepackage{fullpage}
\usepackage{amsthm}
\usepackage{amssymb}
\usepackage{mathtools}
\usepackage{authblk}
\usepackage{thm-restate}
\usepackage{bbm}
\usepackage[linesnumbered,ruled,procnumbered]{algorithm2e}
\usepackage[colorlinks,linkcolor=blue,citecolor=blue,urlcolor=magenta,filecolor=cyan]{hyperref}
\usepackage[capitalize]{cleveref}
\usepackage{capt-of}
\usepackage{hyperref}

\usepackage{tablefootnote} 

\usepackage{fancyhdr}

\newtheorem{theorem}{Theorem}[section]
\newtheorem{lemma}[theorem]{Lemma}
\newtheorem{definition}[theorem]{Definition}

\newtheorem{corollary}[theorem]{Corollary}
\newtheorem{conjecture}[theorem]{Conjecture}

\newtheorem{remark}[theorem]{Remark}

\crefname{procedure}{Procedure}{Procedures}
\SetAlgoProcName{Procedure}{Procedure}

\newcommand{\unweightedmaxcut}[0]{$\mathrm{unweighted}$ $\mathrm{Max}$-$\mathrm{Cut}$\xspace}

\newcommand{\weightedmaxcut}[0]{$\mathrm{Min}$-$\mathrm{UnCut}$\xspace}
\newcommand{\maxcutforreal}[0]{$\mathrm{Max}$-$\mathrm{Cut}$\xspace}
\newcommand{\gapweightedmaxcut}[0]{$(1 - \varepsilon, 1 - \varepsilon^c) \text{-} \mathrm{gap}$ \maxcutforreal}
\newcommand{\threesat}[0]{$3 \text{-} \mathrm{SAT}$\xspace}
\newcommand{\ksat}[0]{$k \text{-} \mathrm{SAT}$\xspace}
\newcommand{\cvptwo}[0]{$\mathrm{CVP}_2$\xspace}
\newcommand{\cvpsubp}[0]{$\mathrm{CVP}_p$\xspace}
\newcommand{\gammacvptwo}[0]{$\gamma \text{-}$\cvptwo}
\newcommand{\gammacvp}[0]{$\gamma \text{-}$\cvpsubp}
\newcommand{\uniquelabel}[0]{$\mathrm{UNIQUE \text{-} LABEL}$\xspace}
\newcommand{\weightedmaxtwosat}[0]{\ensuremath{\mathrm{Max}\text{-}2\text{-}\mathrm{SAT}}\xspace}
\newcommand{\unweightedmaxtwosat}[0]{\ensuremath{\mathrm{unweighted}\text{ }\weightedmaxtwosat}\xspace}
\newcommand{\uniquegames}[0]{$\mathrm{Unique\text{ }Games}$\xspace}
\newcommand{\twontwocvpip}[0]{$2^{n^2}$-$\mathrm{CVP}^{\mathrm{IP}}$}
\newcommand{\orlang}[1]{$\mathrm{OR}(\text{#1})$\xspace}

\newcommand{\gammaANN}[0]{$\gamma \text{-} \mathrm{ANN}$\xspace}
\newcommand{\gammacvpunspec}[0]{$\gamma \text{-} \mathrm{CVP}$\xspace}
\newcommand{\binarycvp}[1][p]{$\mathrm{CVP}^{\{0,1\}}_{#1}$\xspace}
\newcommand{\gammabinarycvp}[1][p]{$\gamma \text{-}$\binarycvp[#1]}

\newcommand{\maxtwolintwo}[0]{\textup{Max-2-Lin(2)}\xspace}
\newcommand{\esgap}[0]{$(\varepsilon, \varsigma)\text{-}\mathrm{gap}$\xspace}
\newcommand{\minusesgap}[0]{$(1-\varepsilon, 1-\varsigma)\text{-}\mathrm{gap}$\xspace}
\newcommand{\eecgap}[0]{$(1-\varepsilon, 1-\varepsilon^c)\text{-}\mathrm{gap}$\xspace}
\newcommand{\esgamma}[0]{$\sqrt[p]{\varsigma/\varepsilon}$\xspace}

\newcommand{\tmop}[1]{\ensuremath{\operatorname{#1}}}
\newcommand{\tmtextit}[1]{\text{{\itshape{#1}}}}

\newcommand{\timecompgammacvp}[0]{\ensuremath{T_{\tmop{CVP}, \gamma} (n)}\xspace}
\newcommand{\timecompesgapmc}[0]{\ensuremath{T_{\tmop{MC}, \varepsilon, \varsigma} (n)}\xspace}

\title{On the (Classical and Quantum) Fine-Grained Complexity of Approximate CVP and Max-Cut}

\author{Jeremy Ahrens Huang}
\author{Young Kun Ko}
\author{Chunhao Wang}
\affil[]{Department of Computer Science and Engineering, Pennsylvania State University}
\affil[]{Email: \{jeremyah,ykko,cwang\}@psu.edu}

\date{}

\begin{document}


\maketitle

\begin{abstract}

We show a linear-size reduction from gap Max-2-Lin(2) (a generalization of the approximate Maximum Cut, or gap $\mathrm{Max}$-$\mathrm{Cut}$, problem) to $\gamma\text{-}\mathrm{CVP}_p$ for $\gamma = \mathrm{O}(1)$ and finite $p \geq 1$, as well as a no-go theorem against poly-sized non-adaptive quantum reductions from $k$-$\mathrm{SAT}$ to $\mathrm{CVP}_2$. This implies three headline results:

(i) Faster algorithms for $\gamma\text{-}\mathrm{CVP}_p$ are also faster algorithms for Max-2-Lin(2) and Max-Cut. Depending on the approximation regime, even a $2^{0.78n}$-time or $2^{0.3n}$-time algorithm would improve upon the state-of-the-art algorithm such as Williams' 2004 algorithm [\textit{Theoretical Computer Science} 2005] or Arora, Barak, and Steurer's 2010 algorithm [$\textit{Journal of the ACM}$ 2015]. This provides evidence that $\gamma\text{-}\mathrm{CVP}_p$ for $\gamma = \mathrm{O}(1)$ requires exponential time, improving upon the previous exponential lower-bound for $\gamma\text{-}\mathrm{CVP}_2$ with $\gamma < 3$ by Bennett, Golovnev, and Stephens-Davidowitz [$\textit{FOCS}$ 2017].

(ii) A new almost $2^{(1/2 + \varepsilon/4\varsigma + o(1)) n}$-time classical algorithm and a new almost $2^{(1/3 + \varepsilon/6\varsigma + o(1)) n}$-time quantum algorithm for $(1-\varepsilon, 1-\varsigma)$-gap Max-Cut. This algorithm is faster than the algorithm of Arora, Barak and Steuer [$\textit{Journal of the ACM}$ 2015], as well as the algorithm of Williams [$\textit{Theoretical Computer Science}$ 2004], 
and the algorithm of Manurangsi and Trevisan [\textit{APPROX/RANDOM} 2018] 
when $c_0 \varepsilon < \varsigma < c_1 \varepsilon$ for some constants $c_0, c_1$. 

(iii) If the Quantum Strong Exponential Time Hypothesis (QSETH) can be used to show a $2^{\delta n}$-time lower-bound for $\mathrm{Max}$-$\mathrm{Cut}$, Max-2-Lin(2), or $\mathrm{CVP}_2$ for any constant $\delta > 0$, it must be via an adaptive quantum reduction unless $\mathrm{NP} \subseteq \mathrm{pr}\text{-}\mathrm{QSZK}$. This illuminates some difficulties in characterizing the hardness of approximate constraint satisfaction problems and shows that the post-quantum security of lattice-based cryptography likely cannot be supported by QSETH. This result complements the no-go results of Aggarwal and Kumar [$\textit{FOCS}$ 2023], who showed that the classical security of lattice-based cryptography likely cannot be supported by the classical Strong Exponential Time Hypothesis (SETH).


\end{abstract}

\newpage

\setcounter{tocdepth}{2}
\tableofcontents

\newpage

\section{Introduction}


The Approximate Closest Vector Problem (\gammacvpunspec) is a hard approximate problem whose asymptotic time complexity \timecompgammacvp is not well understood for approximation parameters $3 \leqslant \gamma < 2^n$. The Approximate Minimum Un-Cut Problem (\esgap \weightedmaxcut), also known as the Approximate Maximum Cut Problem, is another hard approximate problem whose asymptotic time complexity \timecompesgapmc is not well understood for approximation parameters $\varepsilon < \varsigma \leqslant \varepsilon^2$. We show a reduction from \esgap \weightedmaxcut to \gammacvpunspec with a remarkable combination of properties: the reduction precisely preserves both the size and the approximation factor of any instance given to it. By possessing both properties at once, our reduction establishes a tight relationship between the time complexity of \gammacvpunspec 
and of \esgap \weightedmaxcut 
in the approximation regimes of interest. As we will explain in \cref{subsec:technical-difficulties}, finding a reduction possessing both properties at once was a longstanding technical hurdle in the way of establishing such a relationship.

Specifically, our main theorem (\cref{thm:weighted-main-theorem}) states that there is a reduction from \esgap \weightedmaxcut of size $n$ to \esgamma-\cvpsubp of size $n$ (where $p \geqslant 1$ is a constant parameter from CVP). This implies that

\vspace*{-5mm}
\begin{align*}
    \timecompesgapmc &\leqslant \timecompgammacvp
    \intertext{whenever the following (simplified) inequality between the approximation parameters is satisfied:}
    \gamma &\leqslant \sqrt[p]{\varsigma / \varepsilon}.
\end{align*}

We use \cref{thm:weighted-main-theorem} to make three major advancements in the time complexity of CVP and \weightedmaxcut:

First, \cref{thm:weighted-main-theorem} directly gives new exponential-time conditional lower bounds (CLBs) on \gammacvpunspec for $\gamma = \mathrm{O}(1)$ in both the classical and quantum settings. These CLBs are a win-win-win situation: either there are very strong lower bounds on \gammacvpunspec, or there are extremely exciting new algorithms for \esgap \weightedmaxcut, or there are both strong lower bounds on \gammacvpunspec and exciting new algorithms for \esgap \weightedmaxcut. This win-win-win situation is the defining property of fine-grained reductions (see \cref{subsec:reduction-size-general} and \cref{sec:fine-grained-complexity} more on fine-grained complexity).
Previous exponential-time CLBs were only for CVP with $\gamma < 3$ or $\gamma = 1+\mathrm{O}(1/\mathrm{poly}(n))$ in the classical setting from the fine-grained reductions of \cite{BGS17, ABGS21}. 



Second, in combination with a new quantum algorithm presented in this paper and a classical algorithm from \cite{KS20} for special instances of \gammacvpunspec, \cref{thm:weighted-main-theorem} implies two new algorithms for \esgap \weightedmaxcut, one classical and one quantum, which we present in \cref{thm:classical-max-cut-alg} and \cref{thm:quantum-max-cut-alg}. Our classical algorithm is faster than both the fastest exact algorithm \cite{Wil05} and the previous fastest approximation algorithm \cite{MT18}\footnote{Note that a \weightedmaxcut algorithm is listed as an open question in the published version of this article and is present in the arXiv version which was submitted later.} for all instances where $c_1 \varepsilon \geqslant \varsigma > c_0 \varepsilon$ for some constants $c_1, c_0$. Our quantum algorithm is the first to beat Grover search for \esgap \weightedmaxcut!

Third, in combination with a new no-go theorem against \emph{quantum} reductions from \ksat to CVP presented in this paper (\cref{thm:quantum-cvp2-incompressible}) and the no-go theorems against classical reductions in \cite{AK23}, we show that time-complexity lower bounds conditioned on \ksat for \weightedmaxcut, including exact weighted \weightedmaxcut, are unlikely in \emph{both} the quantum and classical settings, assuming no complexity-theoretic disasters occur, for reasons we will explain in \cref{subsec:results-quant-cvp-mc-nogo}. This result could help explain why some desired results in the complexity of approximate problems, such as a proof of the Unique Games Conjecture, have been difficult to obtain. 


Next we will define \gammacvpunspec and \esgap \weightedmaxcut, state our main theorem, and then we will explain each advancement and its related work in more detail.



\subsection{\gammacvpunspec, \weightedmaxcut, and our main theorem}

First we will define \gammacvpunspec. Given a lattice $\mathcal{L}$, a target vector $\bm{t}$, and a distance $r$, \gammacvpunspec is the problem of deciding if $\mathcal{L}$ is $r$ close to $\bm{t}$. If $\mathcal{L}$ is not $r$ close to $\bm{t}$, then it is promised to be more than $\gamma r$ far from $\bm{t}$. A lattice $\mathcal{L} \subset \mathbb{R}^d$ is a set of all the linear combinations of linearly independent basis vectors $\bm{b}^{(1)}, \ldots, \bm{b}^{(n)} \in \mathbb{R}^d$ with integer coefficients,
\begin{align*}
  \mathcal{L}=\mathcal{L} (\bm{b}^{(1)}, \ldots, \bm{b}^{(n)}) \coloneqq
   \left\{ \sum_{i = 1}^n c_i  \bm{b}^{(i)} | c_i \in \mathbb{Z} \right\} .
\end{align*}
Here $n$ is known as the rank of the lattice $\mathcal{L}$ and $d$ is known as the ambient dimension of the lattice $\mathcal{L}$; they are main parameters governing the input size of \gammacvpunspec. The usual way we specify a lattice $\mathcal{L}$ in the input to \gammacvpunspec is by a choice of basis vectors for $\mathcal{L}$, which we represent as a matrix $B \coloneqq \left(\begin{smallmatrix} | &  & |\\ \bm{b}^{(1)} & \cdots & \bm{b}^{(n)}\\ | &  & | \end{smallmatrix}\right)$ that relates every lattice point $\bm{v} \in \mathcal{L}(B)$ to its coordinates $\bm{y} \in \mathbb{Z}^n$ by $\bm{v} = B \bm{y}$. The distance $\mathrm{dist}(\mathcal{L}, \bm{t})$ from a lattice $\mathcal{L}$ to a target $\bm{t}$ is given by the distance of the closest point in $\mathcal{L}$ to $\bm{t}$ in the $\ell_p$-norm. From now on we will call the problem \gammacvp and the distance function $\mathrm{dist}_p$ to emphasize that different $\ell_p$-norms give us different problems. \gammacvp is known to be NP-complete for $\gamma \leqslant n^{c/ \log \log n}$ (where $0 < c < \frac{1}{2}$ is a constant) \cite{DKRS03}, to NOT be NP-complete (unless P=NP) for all $\gamma > \sqrt{n}$ \cite{AR05}, and to be in P for $\gamma \geqslant 2^n$ \cite{SCH87}.


Next we will define \esgap \weightedmaxcut. Given a weighted list of pairs between binary variables and two constants $0 \leqslant \varepsilon < \varsigma \leqslant 1$, \esgap \weightedmaxcut is the problem of deciding if there is an assignment to the variables such that at most a $\varepsilon$-fraction (by weight) of the variable pairs have equal values. If there is no such assignment, then it is promised that for every possible variable assignment at least a $\varsigma$-fraction of pairs have equal values. This makes \esgap \weightedmaxcut a weighted 2-CSP (constraint satisfaction problem) where either at most an $\varepsilon$-fraction of the constraints are unsatisfiable or at least a $\varsigma$-fraction of the constraints are unsatisfiable. \weightedmaxcut is a long-studied problem—the exact case was on Karp's original list of 21 NP-complete problems \cite{Kar72}. \esgap \weightedmaxcut is known to be NP-hard for constant $0 < \varepsilon < 1$ and $\varsigma \leqslant \sqrt{\varepsilon}$ under a widely accepted conjecture about the hardness of approximation \cite{KKMO07}, but in P when $\varepsilon \leqslant \frac{1}{\log n}$ or when $\varsigma > C \sqrt{\varepsilon}$ (for some constant $C$) \cite{CMM06a, CMM06b}. Note that CSPs are sometimes framed in terms of maximizing satisfied constraints instead of minimizing unsatisfied constraints, in which case the gap is given as the the fraction of satisfied constraints. For example, \minusesgap \maxcutforreal is an alternate name which refers to same problem as \esgap \weightedmaxcut. 

Our main theorem actually works with a slight generalization of \esgap \weightedmaxcut known as \minusesgap \maxtwolintwo, which allows each constraint to specify whether the variables should be equal or unequal. Now we are ready to state our main theorem, which is proven in \cref{sec:main-reduction}:
\begin{restatable}[Main theorem]{theorem}{restatemaintheorem}
    \label{thm:weighted-main-theorem}
    There is a linear-time classical reduction and a linear-time quantum reduction from \minusesgap \maxtwolintwo to \esgamma-\cvpsubp with $p \in [1,\infty)$ that uses at most one basis vector per variable. 
\end{restatable}

This reduction has the remarkable property that, for any fixed instance size, the approximation factor $\gamma$ of the output instance increases \emph{without limit} as the gap of the input instance increases. This means our reduction can produce \gammacvp instances of any size with any approximation factor! Previous reductions to \gammacvp had some inherent upper limit on the $\gamma$ of the instances they can produce with $n$ basis vectors for each $n$. See \cref{subsec:technical-difficulties} for a summary of the difficulties involved in overcoming this longstanding limitation. Note also that our reduction maps \minusesgap \maxtwolintwo instances with $n$ variables to \esgamma-\cvpsubp instances with \emph{exactly} $n$ basis vectors; this fact is critical for the tightness of our conditional time bounds on \gammacvp, which we will describe next. After that, we will describe our two other major advancements and their related works. 


\subsection{Our conditional lower bound for \gammacvp and related works}
\label{subsec:prev-results}
\subsubsection{Previous work}


The way that we show a conditional lower bound (CLB) for the time-complexity of \gammacvp is via a reduction that links the time-complexity of \gammacvp to the time complexity of a well-studied problem. This gives us a win-win situation: the lower bound on \gammacvp conditioned on an assumption about the well-studied problem holds unless there is a breakthrough on the well-studied problem. The assumption is necessary since without it we'd be proving that $\mathrm{P} \neq \mathrm{NP}$. Reductions from well-studied problems to \gammacvp have a long history: \cite{Emd81} showed a reduction to exact \cvpsubp and a sequence of works including \cite{ABSS97, Din02, DKRS03, HR07}\footnote{Note that while some of these reductions as written are to a slightly different problem $\gamma$-SVP, by \cite{GMSS99} they are also reductions to \gammacvp and the same principles and difficulties apply.} gave reductions to \gammacvp for greater $\gamma$, with the greatest approximation factor ($\gamma = n^{c / \log \log n}$ for some constant $c > 0$) first achieved by \cite{DKRS03}. These reductions showed that \gammacvp is NP-hard for $\gamma \leqslant n^{c / \log \log n}$. 

However, these reductions only provide a very weak link from the time complexity of their well-studied problems to the time complexity of \gammacvp, because they produce instances of \gammacvp with too many basis vectors. A strong link would require a reduction that uses $n$ basis vectors or less—even using $2n$ basis vectors would meaningfully weaken the link (see \cref{subsec:reduction-size-general} for why). These works use $O(n^k)$ many basis vectors for some (often large and unspecified) constants $k>1$, and the reductions in \cite{ABSS97}\footnote{If reading \cite{ABSS97}, be aware that the problem they call "Label Cover" is very different from (and much easier than) the problem that Label Cover usually refers to today (\cite{MR08}).} for \gammacvp with super-constant $\gamma$ even use super-polynomially many basis vectors! So at best, these results can establish a $2^{\sqrt[k]{n}}$-time conditional lower bound on \gammacvp. For practical values of $n$ these bounds can be even weaker than they appear: at just $k=33$, the inequality $2^{\sqrt[k]{n}} < n$ holds true for $1 < n \leqslant 10^{80}$ (that's the number of atoms in the visible universe).


The fine-grained complexity of \gammacvptwo has recently gained the interest of cryptographers because breaking the lattice-based post-quantum cryptosystems \cite{Reg09, Pei16, GMSS99} currently in widespread use \cite{Ste24} can be done by solving instances of \gammacvptwo. A small change in its quantum time complexity (e.g. from $2^{0.5n}$ to $2^{0.05n}$) can make cracking a secret meaningfully faster, and the previously discussed reductions can give very little evidence for or against such a change. This motivated \cite{BGS17} to show the first result in the fine-grained complexity of \gammacvp: a reduction from \minusesgap \unweightedmaxtwosat or \unweightedmaxcut on $n$ variables to \gammacvp on $n$ basis vectors for $p \in [1, \infty)$ where 
\begin{align}
    \label{eqn:bgs-gamma}
    \gamma = \sqrt[p]{\frac{1 + \varsigma (3^p - 1)}{1 + \varepsilon (3^p - 1)}} < 3.
\end{align}
\minusesgap \unweightedmaxtwosat is a 2-CSP like \minusesgap \unweightedmaxcut, but with constraints that are satisfied when at \emph{least} one variable is 1 instead of when \emph{exactly} one of the two variables is 1. \cite{BGS17} also showed another reduction which was improved by \cite{ABGS21} to go from \ksat on $n$ variables to almost exact \cvpsubp ($\gamma = 1 + 1/\mathrm{poly}(n)$) on $n$ basis vectors, but only for non-even $p$. From now on we will refer to the former reduction as the \cite{BGS17} reduction and the latter reduction as the \cite{ABGS21} reduction for the convenience of the reader. \cite{BGS17} also gave a reduction from exact \weightedmaxcut on $n$ variables to exact \cvpsubp on $n$ basis vectors. 

While there have been no better reductions to \gammacvp with even $p$ or larger $\gamma$ since \cite{BGS17} and \cite{ABGS21}, there is a good reason why. \cite{ABGS21} and \cite{AK23} found that, unless the polynomial hierarchy collapses, any classical reduction to \gammacvp with even $p$ or larger $\gamma$ will likely need to both be from a problem dissimilar to \ksat and use different techniques from the \cite{BGS17} and \cite{ABGS21} reductions. We also find similar restrictions against such quantum reductions in \cref{subsec:results-quant-cvp-mc-nogo}. This is why our reduction is from \maxtwolintwo and uses different geometric gadgets to implement the \maxtwolintwo constraints. 

We summarize the CLBs on \gammacvptwo in \cref{fig:fg-cvpvgamma}. 

\begin{figure}
  \centering
  \includegraphics[width=1\textwidth]{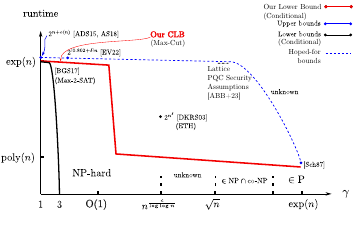}
  \vspace*{-10mm}
  \caption{(Color online.) The classical time-complexity of \gammacvptwo as a function of the approximation factor $\gamma$. 
  Conditional lower bounds: 
  (solid black line) \cite{BGS17} gives an exponential-time lower bound from $\gamma = 1$ to $\gamma \approx 1.5$, which quickly transitions to a $O(1)$-time lower bound and ends at $\gamma = 3$. This corresponds to the transition into the $O(1)$-regime of \minusesgap \unweightedmaxtwosat as $\varsigma$ tends towards 1. $\gamma = 3$ corresponds to $\varsigma = 1$, which is equivalent to deciding if there is a variable assignment that satisfies at least one constraint (there always is). 
  (solid red line) We give an exponential-time lower bound from $\gamma = 1$ through $\gamma = O(1)$, which then transitions to a polynomial-time lower bound that continues past $\gamma = \mathrm{exp}(n)$. This corresponds to the transition away from the exponential-time regime of \esgap \weightedmaxcut which begins as the approximation ratio $\varsigma/\varepsilon$ exceeds $O(1)$, and the transition into the $\mathrm{poly}(n)$-time regime of \esgap \weightedmaxcut which starts when $\varsigma$ increases past $\sqrt{\varepsilon}$ or $\varepsilon$ decreases past $1/\log n$. 
  (black dot) \cite{DKRS03} gives a $2^{n^\delta}$-time bound when $\gamma = n^{c/\log \log n}$. 
  (dashed black line) lattice based post-quantum cryptosystems assume an exponential-time lower bound for $\gamma = \mathrm{poly}(n)$. The line is intended to correspond to the regimes listed in \cite{ABB+23}. 
  Upper bounds, left to right: \cite{ADS15,AS18} give a $2^{n + o(n)}$-time algorithm for $\gamma$ near 1. \cite{EV22} gives an exponential-time algorithm for $\gamma$ less than some unspecified constant. \cite{SCH87} gives a polynomial-time algorithm for exponential $\gamma$.}
  \label{fig:fg-cvpvgamma}
\end{figure}

\subsubsection{Our contribution}

We show a reduction from \minusesgap \maxtwolintwo on $n$ variables to \gammacvp with $n$ basis vectors for $p \in [1, \infty)$ with $\gamma = $ \esgamma in both the quantum and classical settings. This reduction is described in \cref{sec:main-reduction}. It allows us to give the general CLB on \gammacvp in \cref{cor:equal-time} which follows directly from \cref{thm:weighted-main-theorem}. It implies that if one believes a $2^{\delta n}$-time lower bound for some specific constant $\delta > 0$ for \minusesgap \maxtwolintwo or \esgap \weightedmaxcut, then one must also believe a $2^{\delta n}$-time lower bound for \esgamma-\cvpsubp. 

\begin{restatable}{corollary}{restateequaltimecorollary} \label{cor:equal-time}
    Any $O(f(n))$-time classical algorithm for \esgamma-\cvpsubp, where $f(n)$ is lower-bounded by the dimension of the lattice and $p \in [1,\infty)$, is an $O(f(n))$-time algorithm for \minusesgap \maxtwolintwo. The same holds for quantum algorithms.
\end{restatable}

Our reduction makes several improvements over the previous state-of-the-art fine-grained reduction to \gammacvp by \cite{BGS17}: it considers the quantum setting in addition to the classical setting, it works for both weighted and unweighted input instances, and there is no hard limit on how large $\gamma$ is. Using our reduction, we can give interesting CLBs on \gammacvp for approximation factors up to $\gamma = O(1)$.

The increase in the range of approximation factors our reduction works for is very significant. Not only does this tell us something new about the time complexity of many \gammacvp problems that we did not previously have good CLBs for (since \cite{BGS17} only works for $\gamma < 3$), but it also applies a much broader scope of existing and future research on lattice problems to the gap CSPs \weightedmaxcut and \maxtwolintwo. For example, while our reduction gives us new faster classical and quantum algorithms for \esgap \weightedmaxcut and \maxtwolintwo in \cref{subsec:results-mc-algs} using algorithms for general $O(1)$-\cvpsubp these improvements cannot be applied to \minusesgap \unweightedmaxtwosat because the \cite{BGS17} reduction does not work for general $O(1)$-\cvpsubp.


In \cref{tab:bgs-headtohead} we compare our lower bounds to those given by \cite{BGS17} under a stronger version of the above conditions and the conditions used in \cite{BGS17}. The conditions are stronger because we use \unweightedmaxcut instead of \maxtwolintwo since \cite{BGS17} wouldn't work with \minusesgap \maxtwolintwo. In \cref{tab:bounds} we show the maximum possible $\gamma$ with interesting time bounds for some other CLBs.

\begin{table}[h]
\centering
\begin{tabular}{|c|c|c|c|c|c|c|c|}
\hline
\multicolumn{4}{|c|}{Condition} & \multicolumn{2}{|c|}{Our CLB} & \multicolumn{2}{|c|}{\cite{BGS17}'s CLB} \\
\hline
Name & $\varepsilon$ & $\varsigma$ & Time & $\gamma$ & Time & $\gamma$ & Time \\
\hline
\unweightedmaxcut & $O(1)$ & $\sqrt{\varepsilon}$ & $2^{\Omega(n)}$ & $O(1)$ & $2^{\Omega(n)}$ & $\sqrt{2}$ & $2^{\Omega(n)}$ \\
\hline
\cite{Wil05} Optimal $\dagger$ $\star$ & O(1) & $c_0 \varepsilon$ & $2^{\omega n /3}$ & $\sqrt{c_0}$ & $2^{\omega n /3}$ & $\min(\sqrt{2}, \sqrt{c_0})$ & $2^{\omega n /3}$ \\
Gap-ETH\tablefootnote{Gap-ETH can be applied to \minusesgap \unweightedmaxtwosat and \esgap \unweightedmaxcut via the first two reductions in \cite{GJS76}. However, these reductions restrict the size of the gap of their output instances. Since the resulting $\gamma$ is so small, we decided that, as in \cite{BGS17}, the effort required to calculate it exactly is not worthwhile.} $\dagger$ & 0 & $O(1)$ & $2^{\Omega(n)}$ & $<1.115$ & $2^{\Omega(n)}$ & $<1.115$ & $2^{\Omega(n)}$  \\
\hline
\end{tabular}
\caption{Comparison of \cite{BGS17}'s and our CLBs for $\gamma$-approximate \cvptwo, demonstrating that our reduction allows many lower bound conditions to apply to a much wider range of \gammacvp problems than before. Note that while we list \unweightedmaxcut here to allow comparison with \cite{BGS17}, our CLBs give the same time bounds under the much weaker condition that the equivalent gap regime for \minusesgap \maxtwolintwo (a weighted problem) requires $2^{\Omega(n)}$ time. $\dagger$ These are the conditions used in \cite{BGS17} and assume a classical setting. $\star$ \cite{Wil05} is only optimal when $\varsigma \leqslant c_0 \varepsilon$ for some constant $c_0$; thereafter the algorithm presented in \cref{subsec:results-mc-algs} is faster. The time complexity of \cite{Wil05} depends on the matrix multiplication constant $2 \leqslant \omega$.}
\label{tab:bgs-headtohead}
\end{table}

\begin{table}[]
\centering
\begin{tabular}{|c|c|c|c|c|c|}
\hline
Result & Condition & Lower Bound & Maximum $\gamma$ & p  \\
\hline
{\bf Our Result} & gap \maxtwolintwo exponential & $2^{\Omega(n)}$ & $O(1)$ & finite \\
\cite{BGS17} & gap Max-2-SAT exponential & $2^{\Omega(n)}$ & $<3$ & finite \\
\cite{ABGS21} & gap-SETH\tablefootnote{a gap variant of SETH} & $2^{n}$ & $<1 + 1/\mathrm{poly}(n)$ & non-even \\
\cite{DKRS03}\tablefootnote{$0<\delta,c<1/2$} & ETH & $2^{n^\delta}$ & $O(n^{c/\log \log n})$ & finite \\
\hline
\end{tabular}
\caption{Some conditional lower bound results for $\gamma$-approximate \cvpsubp, presented so as to maximize the value of $\gamma$ while still maintaining interesting time lower bounds.}
\label{tab:bounds}
\end{table}

The main reason our CLB applies to \gammacvp with large approximation factors while also giving tight time bounds is that our reduction uses a novel geometric gadget to represent \maxtwolintwo constraints (described in \cref{sec:main-reduction}) that only uses two basis vectors, while allowing the ratio between the "cost" when a constraint is unsatisfied and the cost when a constraint is satisfied to be fully parameterized. We label this parameter as $\iota$, and our gadget forces the \gammacvp instance to "pay" a cost of $\iota$ when a constraint is unsatisfied instead of a cost of $1$ when it is satisfied. In the previous state of the art reduction due to \cite{BGS17} this ratio is 3, and it cannot be increased as the family of geometric gadgets they use is inherently limited to a small constant dependent on the choice of norm and the number of variables in a constraint. 

 The value $\gamma$ given by the unweighted version of our reduction, shown in \cref{eqn:our-unweighted-gamma}, helps illustrate the significance of the ratio $\iota$.
\begin{align}
    \label{eqn:our-unweighted-gamma}
    \gamma = \sqrt[p]{ \frac{1 + \varsigma(\iota^p-1)}{1 + \varepsilon(\iota^p-1)}}
\end{align}
Comparing it to the value of $\gamma$ for \cite{BGS17} given in \cref{eqn:bgs-gamma} clearly shows that fixing the ratio at 3 constrains the contribution of the \esgap to the approximation ratio while allowing the ratio to be an arbitrarily large parameter which can depend on $n$ and $p$ allows us to maximize the contribution of the gap to $\gamma$. But it can go further: when the ratio is constant, that places a constant upper bound on the possible values of $\gamma$, even when we make the input problem trivial by setting $\varsigma=1$. But when the ratio is a parameter we can get any arbitrary $\gamma$ we want by varying $\iota$ and setting $\varepsilon \ll \varsigma$ (or even $\varepsilon=0$) and keeping $\varsigma \ll 1$, which gives us a non-trivial, polynomial-time conditional lower bound for \gammacvp for all $\gamma$. This showcases our reduction's unlimited approximation-ratio preserving property, which as far as we can tell is unprecedented in the literature. 

\subsubsection{Subsequent work}

Recently \cite{AK25} showed an algorithm (related to \cite{Wil05}) for exact \cvpsubp which can solve the \cvpsubp instances created by our and \cite{BGS17}'s reductions in the same amount of time as \cite{Wil05} solves their input Max-2-CSP instances. This means that the \cite{Wil05}-based CLBs for exact \cvpsubp shown by us and by \cite{BGS17} are perfectly tight—any reduction to \cvpsubp on fewer basis vectors would result in a faster algorithm for an exact Max-2-CSP than \cite{Wil05}.

That concludes our overview of CLBs for \gammacvp. Next we will discuss our new algorithms for \minusesgap \maxtwolintwo. 

\subsection{Our \weightedmaxcut algorithms}
\label{subsec:results-mc-algs}

Exact \weightedmaxcut was one of Karp's first 21 NP-complete problems \cite{Kar72}. The current best classical algorithm for it is due to \cite{Wil05}, who gave an exponential-sized reduction from \weightedmaxcut and \weightedmaxtwosat to a graph problem in P and a way to solve that problem using matrix multiplication. This means that the time complexity of \cite{Wil05}'s algorithm is given in terms of the matrix multiplication exponent $2 \leqslant \omega$. The complexity of \cite{Wil05}'s algorithm is $O(2^{\omega n /3})$ which is approximately $O(2^{0.79 n})$ using the current best exponent $\omega \approx 2.371$ due to \cite{ADW+24}. The current best algorithm for \esgap \weightedmaxcut with constant approximation factors $\alpha = \varsigma/\varepsilon$ in the conjectured NP-hard regime is $\exp(n/2^{\Omega(\alpha^2)})$-time \cite{MT18}\footnote{Note that the \weightedmaxcut algorithm is absent from the published version of this article and present in the arXiv version which was submitted later.}. \cite{ABS15} give a sub-exponential algorithm for gap \maxtwolintwo in the region where $\varsigma = C \sqrt{\varepsilon}$ for some constant $C$ which gets slower as $C \rightarrow 1$, reaching $O(2^n)$ time when $C=1$. The current best quantum algorithm for exact \weightedmaxcut and \maxtwolintwo is naive search over the variable assignments using Grover's algorithm \cite{Gro96}, which takes $O(2^{n/2})$ time. For more about \esgap \weightedmaxcut, \maxtwolintwo, or the Unique Games Conjecture, see \cref{sec:max-cut} or the survey \cite{Kho10}.



We give a new $O\Bigl(2^{ \bigl( \frac{1}{3} + \frac{2}{6 \gamma^2 - 3} + o(1) \bigr) n}\Bigr)$-time quantum algorithm for a special case of \gammacvptwo and then use our reduction from \cref{thm:weighted-main-theorem} to give an almost $O(2^{n(\frac{1}{3} + \frac{\varepsilon}{6\varsigma} + o(1))})$-time algorithm for \minusesgap \maxtwolintwo (\cref{thm:quantum-max-cut-alg}). This is the first quantum algorithm to improve upon naive search for \minusesgap \maxtwolintwo. We also get an almost $O(2^{n(\frac{1}{2} + \frac{\varepsilon}{4\varsigma} + o(1))})$-time classical algorithm for \minusesgap \maxtwolintwo (\cref{thm:classical-max-cut-alg}) by starting with a $O(2^{ \left( \frac{1}{2} + \frac{1}{4 \gamma^2 - 2} + o(1) \right) n})$-time classical algorithm from \cite{KS20}. 

Our classical algorithm is faster than \cite{ABS15} when $\varsigma \leqslant \sqrt{\varepsilon}$ and it is faster than \cite{Wil05} when $c_0 \varepsilon < \varsigma$ for some constant $c_0 \approx 3$ which depends on the sub-constant factors in the exponent and on $\omega$. Due to large implicit constants in \cite{MT18}, our algorithm is faster when $\varsigma < c_1 \varepsilon$ for some large constant $c_1$. That makes it the best known classical algorithm for \minusesgap \maxtwolintwo or \esgap \weightedmaxcut when $c_0 \varepsilon < \varsigma < c_1 \varepsilon$. Our quantum algorithm is faster than our classical algorithm, and is faster than naive search when $\varsigma > c_q \varepsilon$ for some constant $c_q$ dependent on the sub-constant factors in the exponent, so it is the best known algorithm for $c_q \varepsilon < \varsigma < c_1 \varepsilon$.


When $\varepsilon = o(1)$, our quantum and classical algorithms have time complexities $O\Bigl(2^{ \bigl( \frac{1}{3}+ o(1) \bigr) n}\Bigr)$ and $O\Bigl(2^{ \bigl( \frac{1}{2}+ o(1) \bigr) n}\Bigr)$ respectively. So our algorithms are also faster than all other known algorithms when $\varepsilon = 1/f(n)$ for $f(n) \leqslant \log n$ and $\varsigma \leqslant \sqrt{\varepsilon}$. 

Next we will give an overview of our last major advancement: the no-go results against quantum reductions from \ksat to \cvptwo or \weightedmaxcut.


\subsection{Our no-go theorems}
\label{subsec:results-quant-cvp-mc-nogo}

After \cite{BGS17}, attempts were made to circumvent the barriers against extending those reductions so that they are able to produce \gammacvp instances with greater approximation factors $\gamma$ or with even $p$. Although those attempts failed, they also found good reasons for why the barriers are likely insurmountable and created a new area of fine-grained complexity in the process. The focus was first on extending the reduction from \ksat to \cvptwo; \cite{ABGS21} found that the geometric gadget used in the reduction ($(p,k)$-isolating parallelepipeds) to represent \ksat clauses does not exist for even $\ell_p$-norms when $k > p$, then \cite{AK23} found that (most) reductions from \ksat on $n$ variables to \cvptwo on $\mathrm{poly}(n)$ basis vectors would collapse the polynomial hierarchy! This also helped explain the difficulty in extending the reduction to work for larger $\gamma$: \cite{AK23} found that since for every $p$ there is a constant $\gamma_p$ such that $\gamma_p$-\cvpsubp can be reduced to \cvptwo (a result due to \cite{EV22}), any reduction which produces \gammacvp instances with $\gamma \geqslant \gamma_p$ could still be subject to no-go theorems for \cvptwo.

Although \unweightedmaxtwosat isn't \ksat, these no-go results can also help explain the limitations of the \cite{BGS17} reduction from \minusesgap \unweightedmaxtwosat to \gammacvp. If \ksat can't be reduced to \gammacvp with larger $\gamma$, then the geometric gadgets which could be used to represent \ksat must fail for larger $\gamma$ just like how \cite{ABGS21} found that they fail for even $p$. Since the \cite{BGS17} reduction from \unweightedmaxtwosat is reliant on the same geometric gadgets (isolating parallelepipeds) used in their reduction from \ksat, it makes sense that it would face limitations similar to what it would face if it were a reduction from \ksat. 


\cite{AK23} realized that since tight conditional lower bounds (CLBs) hinge on reductions with a specific numerical relationship between instance size of the input problem and the instance size of the output problem (also known as fine-grained reductions, which we explain in more detail in \cref{subsec:prev-results}), their no-go theorems mean that conditions about the time complexity of \ksat like the Strong Exponential Time Hypothesis (SETH) probably can't be applied to \cvptwo.

We show that there are no polynomial-sized, non-adaptive, \emph{quantum} polynomial-time reductions from \ksat to \cvptwo with two-sided error unless there are quantum statistical zero-knowledge proofs for all languages in $\mathrm{NP}$ (\cref{thm:quantum-cvp2-incompressible}). This result shows a substantial barrier against proving the fine-grained complexity of \cvptwo using the Quantum Strong Exponential Time Hypothesis and fills a gap in \cite{AK23}, which only showed no-go results for \emph{classical} reductions. Somewhat counter-intuitively, because we were able to avoid these limitations with our fine-grained reduction from \minusesgap \maxtwolintwo to \gammacvp, we can show that these limitations must apply to quantum fine-grained reductions from \ksat to \maxtwolintwo no matter how complex the weights are (\cref{thm:quantum-no-ksat-to-ug-reduction}). Note that since there is a trivial reduction from \minusesgap \maxtwolintwo, exact \weightedmaxcut, and \esgap \weightedmaxcut to \maxtwolintwo, the no-go theorem applies to them as well. 

Our reduction also allows us to directly rule out classical reductions from \ksat to \weightedmaxcut based on the results \cite{AK23}. In particular, we rule out all fine-grained reductions with one-sided error (\cref{thm:classical-mc-no-turing}) and all non-adaptive fine-grained reductions with two-sided error (\cref{thm:no-randomized-reductions-from-seth}). Surprisingly, these barriers seem to close off the most tempting approaches toward understanding the fine-grained complexity of Max-Cut. However, they might help explain the lack of results in this area. These results could also be derived indirectly from a corollary in \cite{BGS17} and \cite{AK23}, although this seems to have been unnoticed despite its significance for such a well-studied problem.

These results mean that there are now a group of problems with fine-grained reductions to \cvptwo in addition to the existing group of problems with fine-grained reductions from \ksat, and that there (probably) cannot be fine-grained reductions from the latter group to the former group. This suggests that the two groups can be considered as two distinct \emph{fine-grained} complexity classes. Note that these fine-grained complexity classes are independent of regular complexity classes since, as demonstrated by \cite{Wil05}'s algorithm described in \cref{subsec:results-mc-algs}, problems in P and problems that are NP-complete can reside in the same fine-grained complexity class. 

We also show that all ``natural'' reductions from $n$-variable 3-SAT to \weightedmaxcut must use $\frac{4}{3}(n-2)$ vertices (\cref{thm:natural-reductions-ksat-mc-four-thirds}) for a definition of ``natural'' given by \cite{ABGS21}. 

Now we have completed an overview of all of our advancements. Next we will give a brief overview of our reduction.


\subsection{Our Reduction}

We first show an abridged reduction from unweighted Max-Cut to \cvpsubp for $p \in [1,\infty)$ and how the approximation factor for the output \cvpsubp instances from the gap of the input Max-Cut instances, without needing to handle the complexity of the full proof for Max-2-Lin(2). For the full reduction, see \cref{sec:main-reduction}.

Given a Max-Cut instance as an ordered set of unordered pairs $E \subseteq[n]^2$ representing $m$ inequality constraints on $n$ variables $v_1,\dots,v_n$, we produce a basis $B$ with $n$ vectors on $2m$ dimensions and a target $\bm{t} 
\in \mathbb{R}^{2m}$. We represent $B$ as a matrix in $\mathbb{R}^{2m \times n}$ with the basis vectors as its columns. We construct $B$ and $\bm{t}$ as a series of $m$ vertically stacked blocks, where each block is 2 rows tall:

\[
B := \left(\begin{matrix}
    B_{1:2}\\
    B_{3:4}\\
    \vdots\\
    B_{2m-1 : 2m}
\end{matrix}\right), 
\bm{t}:=\left(\begin{matrix}
    \bm{t}_{1:2}\\
    \bm{t}_{3:4}\\
    \vdots\\
    \bm{t}_{2m-1 : 2m}
\end{matrix}\right).
\]
Each pair of blocks (one in $B$ and one in $\bm{t}$) makes up one gadget, and each gadget encodes one constraint from $E$; specifically, we construct the gadget $B_{2k-1 : 2k}, \bm{t}_{2k-1:2k}$ to implement the $k^\text{th}$ constraint of $E$ as follows:
\[
B_{2k-1 : 2k} := \left(\begin{matrix}
   \cdots & -1 & \cdots & 1 & \cdots \\
   \cdots & \iota & \cdots & \iota & \cdots
\end{matrix}\right), 
\bm{t}_{2k-1:2k}:=\left(\begin{matrix}
    0\\
    \iota
\end{matrix}\right),
\]
where only the $i^\text{th}$ and $j^\text{th}$ columns of $B_{2k-1 : 2k}$ are shown with $(i,j):=E_k$, $i<j$, and the remaining columns are filled with 0s. We show the columns in order with column $i$ first. $\iota > 1$ is a parameter we use to ensure that the distance to the target is greater when the constraint is unsatisfied than when the constraint is satisfied. We use two rows per constraint to ensure that the basis vectors for our lattice are linearly independent—otherwise our reduction wouldn't produce valid \gammacvp instances. We visualize the block in \cref{fig:fg-cvp} by plotting the $(2k-1)^\text{th}$ and $2k^\text{th}$ rows of the vectors $B_{\cdot,i}, B_{\cdot,j}, \bm{t},$ and $B_{\cdot,i} + B_{\cdot,j}$. $B_{\cdot,i}$ represents the $i^\text{th}$ column of $B$, which is also the $i^\text{th}$ basis vector.

\begin{center}
  \includegraphics[width=0.5\textwidth]{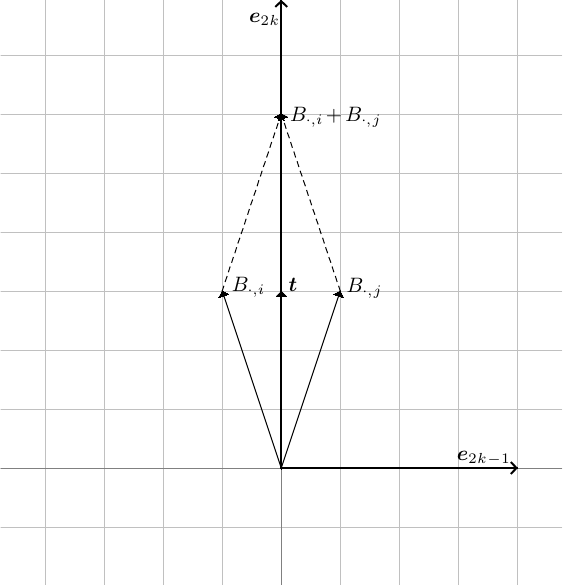}
  \captionof{figure}{visualization using $\iota=3$ of the $k^\text{th}$ gadget, which encodes the constraint $v_i \neq v_j$.}
  \label{fig:fg-cvp}
\end{center}

As you can see from \cref{fig:fg-cvp}, $B_{\cdot, i}$ and $B_{\cdot, i}$ are equidistant from the target $\bm{t}$, $\bm{0}$ and $B_{\cdot, i} + B_{\cdot, j}$ are also equidistant from $\bm{t}$, and $B_{\cdot, i}$ and $B_{\cdot, j}$ are closer to $\bm{t}$. If we interpret the coefficient of $B_{\cdot,i}$ as the assigned value of $v_i$, we find that the gadget matches the preferences of the constraint $v_i \neq v_j$ exactly: it prefers $v_i,v_j \in \{(0,1), (1,0)\}$ over $v_i,v_j \in \{(0,0), (1,1)\}$ and it has no preference between $(0,0)$ and $(1,1)$ and no preference between $(0,1)$ and $(1,0)$. 

\begin{theorem}[Informal]
    The above reduction reduces \esgap \unweightedmaxcut to \esgamma-\cvpsubp. 
\end{theorem}

\begin{proof}[Proof sketch] (See \cref{sec:main-reduction} for the full proof.)
    We first show that the reduction produces a lattice and target with distance at most $(m(1 + \varepsilon(\iota^p-1)))^{1/p}$ when given a set of constraints that is $(1-\varepsilon)$-satisfied by some variable assignment $v_1,\dots,v_n$ (this the completeness condition for our reduction). Consider the lattice point $\bm{u} := \sum_{i=1}^n v_i B_{\cdot,i}$. Then
    \[|\bm{u} - \bm{t}|_{2k-1:2k} = \begin{cases}
        \big(\begin{smallmatrix}0 \\ 1\end{smallmatrix}\big) & \text{the $k^\text{th}$ constraint is satisfied} \\
        \big(\begin{smallmatrix}\iota \\ 0\end{smallmatrix}\big) & \text{the $k^\text{th}$ constraint is not satisfied} \\
    \end{cases}\] by the construction of the $k^{th}$ gadget (where we use $||$ to denote element-wise absolute value), so by summing up the blocks corresponding to satisfied constraints and the blocks corresponding to unsatisfied constraints we get
    \[\mathrm{dist}_p(\bm{u}, \bm{t})^p = m(1-\varepsilon) + m \varepsilon \iota^p = m(1 + \varepsilon(\iota^p-1))\] as claimed (note the power of $p$ on the left hand side). 
    
    Next we show that the reduction produces a lattice and target with distance greater than $(m(1 + \varsigma(\iota^p-1)))^{1/p}$ when given a set of constraints that is less than $(1-\varsigma)$-satisfied by any variable assignment (this is the soundness condition for our reduction). Assume that there is always a lattice point with 0-1 coefficients that is at least as close to the target as every other lattice point (we omit proof of this fact in this sketch). We can construct a variable assignment $v_1,\dots,v_n$ from any lattice point with 0-1 coefficients $\bm{u}$ by choosing $v_i$ to be the coefficient of $B_{\cdot,i}$. Then by the argument given above for completeness the distance from $\bm{u}$ to $\bm{t}$ is given by the number of constraints satisfied by $v_1,\dots,v_n$ as
    \[\mathrm{dist}_p(\bm{u}, \bm{t})^p = (\text{\# of satisfied constraints}) + \iota^p (\text{\# of unsatisfied constraints}).\]
    Since every variable assignment satisfies fewer than $m(1-\varsigma)$ constraints, 
    \[\mathrm{dist}_p(\bm{u}, \bm{t})^p > m(1-\varsigma) + \iota^p m \varsigma = m(1 + \varsigma(\iota^p-1))\]
    for any lattice point $\bm{u}$ with 0-1 coefficients. Applying our assumption applies this bound to all lattice points as claimed. 

    By choosing $r:=(m(1 + \varepsilon(\iota^p-1)))^{1/p}$ and $\gamma r := (m(1 + \varsigma(\iota^p-1)))^{1/p}$ we get that 
    \[\gamma = \frac{(m(\varsigma(\iota^p-1) +1))^{1/p}}{(m(\varepsilon(\iota^p-1) + 1))^{1/p}} 
    = \sqrt[p]{ \frac{1 + \varsigma(\iota^p-1)}{1 + \varepsilon(\iota^p-1)}}
    \approx \sqrt[p]{\varsigma/\varepsilon}\] as required, when $\iota$ is sufficiently large. 
\end{proof}

Next we will give some context on the longstanding technical difficulty our reduction overcame.

\subsection{Difficulty of fine-grained reductions to \gammacvp with large $\gamma$}
\label{subsec:technical-difficulties}

Our reduction overcomes a longstanding technical difficulty for past reductions to lattice problems: producing an instance that has a large approximation factor while maintaining a small number of basis vectors. As you may have noticed in \cref{subsec:prev-results}, all past reductions to \gammacvp have either produced many basis vectors to get large $\gamma$ or produced few basis vectors with small $\gamma$. We will try to give some intuition as to why this double-bind occurred. 

One way to think about how reductions to \gammacvp with large $\gamma$ such as \cite{ABSS97, DKRS03, HR07}\footnote{Again, while some of these reductions as written are to a slightly different problem $\gamma$-SVP, by \cite{GMSS99} they are also reductions to \gammacvp and the principles described here are the same.} work is that they increase the number of dimensions polynomially while also increasing the number basis vectors polynomially. In general, increasing the number of basis vectors produced by a reduction from a CSP to \gammacvp decreases the distance to the target, while increasing the number of dimensions produced by a reduction increases the distance to the target. These reductions construct their additional basis vectors in a systematic manner in which reduces the distance much more if the input instance is good than if it is bad. When combined with the fact that the bad instances have a long distance to reduce from this provides a large approximation factor (remember that the approximation factor is the long distance divided by the short distance). However, we know from \cref{subsec:reduction-size-general} that this strategy doesn't produce fine-grained reductions because it uses too many basis vectors. 

Ideally, we would like to produce a reduction from a problem on $n$ variables to \gammacvp with exactly $n$ basis vectors. Such reductions to exact \cvpsubp were known at least as far back as \cite{ABSS97}, but extending them to work for larger $\gamma$ is difficult.

To describe these difficulties in constructing problems to \gammacvp, it is helpful to have a generalization for the problems we are reducing from. We will use Constraint Satisfaction Problems (CSPs), which neatly generalize all the well-studied problems mentioned in this paper. CSPs are about deciding if there is a combination of variable assignments that is good enough at satisfying a set of constraints (restrictions on variable assignments). An assignment is considered good enough if it satisfies a $1-\varepsilon$ fraction of the constraints; if no assignment is good enough, then it is promised that all assignments can only satisfy at most a $1-\varsigma$ fraction of constraints. The input size is given by the number of variables $n$ and the number of constraints $m$.

One difficulty is as follows: these reductions (including our own) always represent the CSP variables as basis vectors so that they can correspond the integer coefficients of the lattice points to the values of variables in the CSP\footnote{This is in contrast to, for example, \cite{DKRS03} where there were enough basis vectors to have them represent every possible combination of constraint and variable assignment.}. There are infinitely many integer coefficients but only a finite number of possible values for each variable, so to prove that reduction faithfully implements the CSP we must show that the valid integer coefficients that correspond to possible CSP variable values are always further away than the invalid coefficients. This is usually done by adding a ``leashing gadget'': one or more dimensions where the value of basis vectors and target forces lattice points to be far away when a basis vector is assigned an invalid integer coefficient. For example, to force a basis vector to only have valid coefficients 0 and 1 we might set the target to $\alpha$ in the leashing gadget's dimension and the basis vector to $2 \alpha$ for some very large $\alpha$. Since the basis vectors must be allowed at least two valid integer coefficients but the target can only have one value, the leashing gadgets must add some distance to every \gammacvp instance, which we can again call $\alpha$. To prove that the leashing gadget always works, $\alpha$ must exceed the gap between close and far instances of \gammacvp for all the other dimensions combined which caps the approximation factor at 2. 

This specific difficulty was first overcome by \cite{BGS17} which showed a \cvpsubp gadget (2-dimensional isolating parallelepipeds) that they proved will faithfully implement a CSP constraint (2-ary binary conjunction) even when ``unleashed'', breaking the $\gamma = 2$ barrier. However, the double-bind remained, due to another difficulty: the factor of the distances of far and close lattice points on an isolating parallelepiped are inherently limited by a small explicit constant. Even with the most favorable parameters, unsatisfied constraints are limited to 3-times the distance of satisfied constraints, limiting the overall approximation factor to $\gamma < 3$. 

\subsection{Why fine-grained reductions must be small}

\label{subsec:reduction-size-general}

Reductions from a Problem A to a Problem B are commonly used in both
``coarsed-grained'' complexity, to establish that B is in the same complexity
class as A, and ``fine-grained'' complexity, to establish a specific numerical
relationship between the complexities of A and B. In general, a reduction from
A to B is a procedure which solves an instance of A using the solutions to
instances of B. A reduction R from Problem A to Problem B which solves any
size-$n$ instance of A using the solution(s) to $c_R (n)$ many size-$s_R (n)$
instance(s) of $B$ in time $T_R (n)$ naturally gives the following upper-bound
on time complexity of A $T_A$ in terms of the time complexity of B $T_B$:
\begin{eqnarray*}
  T_A (n) & \leqslant & c_R (n) T_B (s_R (n)) + T_R (n) .
\end{eqnarray*}
What makes reductions so useful in complexity theory is that this inequality
is also a \tmtextit{lower-bound} on $T_B$ in terms of $T_A$. We can make this
clearer by rearranging the inequality and introducing a change of variables
$n' = s_R (n)$:
\begin{eqnarray*}
  \frac{T_A (s_R^{- 1} (n')) - T_R (s_R^{- 1} (n'))}{c_R (s_R^{- 1} (n'))} &
  \leqslant & T_B (n') .
\end{eqnarray*}
Note that $s_R$, commonly known as the size of the reduction R, is present in
every term.

A common practice is to give a polynomial-size, polynomial-time, many-to-one
reduction and assume that problem A requires super-polynomial time; plugging
these parameters into our inequality gives \ $T_B (n) \geqslant
\tmop{superpoly} (\tmop{poly}^{- 1} (n))$ which is a super-polynomial
lower-bound on $T_B$ (since $T_R (n) \in o (T_A (n))$ it does not contribute
to the asymptotic complexity of $T_B$). Thus polynomial-size reductions are an
effective way to place Problem B in the same complexity class as Problem A.

While such reductions can answer some coarse-grained complexity questions,
they are too large to convert a good algorithm for B into a good algorithm for
A or a good lower-bound for A into a good lower-bound for B. This is because
the large reduction size converts large changes on one side of the inequality
to small changes on the other side.

Let us demonstrate by a lower-bound example: assume that Problem A requires
$2^{n / 2}$-time and let C be a many-to-one reduction from A to B with size
$2^k n^k$ for some $k > 1$ in polynomial time, and let R also be a reduction
from A to B in polynomial time with size $2 n$. To emphasize the point, let R
also use $n^{10}$ many instances of B instead of just one. Then by our
inequality, reduction C establishes a $\Omega \left( 2^{\sqrt[k]{n} / 4}
\right)$ lower-bound on B while reduction R establishes a $\Omega (2^{n / 4}
n^{- 10})$ lower-bound. Since $k > 1$, it's clear that R gives a much, much
better asymptotic lower-bound than C does. Now imagine that a celebrated
result shows large improvement in the lower-bound for A, from $2^{n / 2}$-time
to $2^n$. How does that change our lower-bound results for B? The bound given
by C improves by only a factor of $2^{\sqrt[k]{n} / 4}$ (to $\Omega \left(
2^{\sqrt[k]{n} / 2} \right)$) while the bound given by R improves by a factor
of $2^{n / 4}$ (to $\Omega (2^{n / 2} n^{- 10})$)!

Next we will demonstrate the important of the exact constant factor of our
linear-sized reduction. Let reduction R stay as before, and consider another
reduction F which is exactly like R except that it has size $n$ instead of $2
n$ and takes time $2^{\sqrt[4]{n}}$ instead of $\tmop{poly} (n)$. Then
reduction F gives a lower bound of
\[ \left( 2^n - 2^{\sqrt[4]{n}} \right) n^{- 10} \in \Omega (2^n n^{- 10}) \]
on B instead of the $\Omega (2^{n / 2} n^{- 10})$ bound given by R! The above
comparisons demonstrate that a reduction which can give good fine-grained
complexity results (a fine-grained reduction) for potentially exponential-time
problems must be a linear-sized reduction, that the constant factor in the
reduction size is more important than the complexity of other aspects of the
reduction, and that the constant factor is ideally 1 or less.

\subsection{Open questions}
We leave the following questions to future research.
\begin{enumerate}
    \item Can the fine-grained hardness of \gammacvptwo be further extended to show that there are no $O(2^{n^{1-\delta}})$-time algorithms with $\delta > 0$ for \gammacvptwo with approximation factors further into the $\mathrm{poly}(\log n)$ regime and even beyond it towards $\sqrt{n}$ regime used by modern cryptography?
    \item Now that we have identified that fine-grained approximation-ratio preserving reductions are achievable, can we create other fine-grained approximation-ratio preserving reductions that help answer questions in the hardness of approximation?
    \item Can the fine-grained hardness results shown in this work also be applied to the approximate Shortest Vector Problem?
    \item Do fine-grained reductions to \gammacvp which utilize the ``full power'' of \gammacvp by using lattice coordinates with many different coefficients instead of just two lead to stronger lower bounds? Or maybe \gammacvp fully characterized by existing reductions, since for $1 < p< \infty$ any line formed by a basis vector can only be closest to the target at two points in the $\ell_p$-norm\footnote{\cite{ABGS21} give a good technical argument for this possibility.}? If the former is true, then hopefully we can use the proof can prove much higher fine-grained lower bounds on \gammacvp. If the latter is true, then we can use existing algorithms to match our current lower bounds. See \cref{sec:lattice-problems} for how \gammacvp problems with different lattice coefficients relate to each other. 
    \item The ``no-go'' results and reductions discussed in this paper suggest that \weightedmaxcut and \gammacvptwo belong together in a certain fine-grained complexity class, and that this class is separate from the much better understood fine-grained complexity class to which \ksat belongs. What is the true fine-grained complexity lower bound for this class? How should we formulate a fine-grained lower bound conjecture for this class, and which problem should the conjecture be formulated for? One candidate is that there is no $O(2^{n^{1-\delta}})$-time Max-Cut algorithm for $\delta > 0$; another candidate in use by some cryptographers is that there is no $O(2^{0.2075 n})$-time algorithm for \gammacvptwo. 
\end{enumerate}

\subsection{Acknowledgements}
We thank Sean Hallgren for valuable discussions and feedback. JAH and CW were supported by a National Science Foundation grant CCF-2238766 (CAREER).

\section{Preliminaries}
\label{sec:prelim}

\subsection{Notation}
Throughout this paper, we use bold-face lower-case Roman letters to denote column vectors, such as $\bm{v}$ and $\bm{t}$. For a vector $\bm{v}$, we use $\bm{v}_j$ to denote its $j$-th entry. The $\ell_p$-norm of a vector $\bm{v}$, denoted $\norm{\bm{v}}_p$, is defined as
\begin{align*}
  \norm{\bm{v}} \coloneqq (| \bm{v}_1 |^p + \cdots + | \bm{v}_d |^p)^{1/p}.
 \end{align*}
We use upper-case Roman letters such as $B$ to denote matrices. We use $B_{i, j}$ to denote the $(i, j)$-entry of $B$ and $B_{\cdot, i}$ to denote the $i$-th column of $B$. Unless otherwise noted, we use $\varepsilon$ to denote the small constant that is used to describe the fraction of unsatisfiable constraints for the YES instances of a CSP (e.g. $1-\varepsilon$ completeness), we use $\varsigma$ (sigma) to denote the small constant that is used to describe the fraction of unsatisfiable constraints for the NO instances of a CSP (e.g. $1-\varsigma$ soundness), and we use $\delta$ to denote constant fractions used in the exponent of exponential time complexities (e.g. $2^{n^\delta}$ or $2^{(1-\delta) n}$). We also sometimes use $c$ to denote a small constant used with $\varepsilon$ as $\varepsilon^c$ to describe $1-\varepsilon^c$ soundness. We reserve $\gamma$ for representing the approximation factor of approximate problems (e.g. \gammabinarycvp, \gammaANN). We use $[n] \coloneqq \{1, 2, \ldots, n\}$. 

Throughout this paper we will work directly with $\mathbb{R}^d$ for readability, without including the formal details of selecting and analyzing the complexity of a suitable finite representation.  

\subsection{Fine-grained complexity}
\label{sec:fine-grained-complexity}

The ultimate goal of complexity research is to prove lower bounds on the computational resources needed to solve computational problems; however, proving lower bounds unconditionally has been difficult in practice. Fine-grained complexity approaches this difficulty by focusing on studying how the lower bounds of different problems relate to each other instead. The main tool of fine-grained complexity is the \emph{fine-grained reduction}, which, informally, is a reduction from a problem $P_\mathrm{from}$ to a problem $P_\mathrm{to}$ that is so small and fast that any improved upper-bound on $P_\mathrm{to}$ is also an improved upper-bound on $P_\mathrm{from}$ (see \cref{subsec:reduction-size-general} for what makes a reduction more or less fine-grained, and see the survey \cite{Vas15} for formal definitions). 

If one is confident about a conjectured lower bound on one problem, then fine-grained reductions from that problem allow one to conjecture lower bounds for other problems with similar confidence. One longstanding conjecture is that the naive search algorithm for \ksat is optimal; this conjecture has been formulated in different ways over the years (for example, in \cite{SH90} and \cite{IP99}) and is now known as the Strong Exponential Time Hypothesis (SETH). This is the go-to conjecture for showing that a problem requires strongly exponential time. Below we reproduce the conjecture as it is formulated in \cite{Vas15}, where the classical Word RAM model with $O(\log n)$ bit words is assumed. 

\begin{conjecture}[SETH]
    For every $\delta > 0$, there exists an integer $k$, such that Satisfiability of $k$-CNF formulas on $n$ variables cannot be solved in $O(2^{(1-\delta) n} \mathrm{poly}(n))$ time in expectation. 
\end{conjecture}

Note that SETH does not state that \ksat requires exactly $2^n$-time for any specific $k$; instead it states that the time complexity of \ksat approaches $\Omega(2^n)$ as $k$ increases. So showing that a problem has a $2^{n}$-time lower bound by fine-grained reduction from \ksat and assuming SETH requires the reduction to be from \ksat for all $k \in \mathbb{N}$. 

Naive search is known to have quadratic speedup in quantum models of computation, so a quantum variation of SETH was required. 

\begin{conjecture}[QSETH \cite{ACL+20}]
    For all $\delta > 0$, there exists some $k \in \mathbb{N}$ such that there is no quantum algorithm solving \ksat with $n$ variables in time $O(2^{n (1-\delta) / 2})$. 
\end{conjecture}

This formulation of QSETH is directly comparable to SETH with just the coefficient in the exponent and the model of computation changed. There is another more general formulation of QSETH for a range of problems due to \cite{BPS19} which simplifies to the ``basic'' formulation given above for the satisfiability of CNF-formulas.

As its name suggests, there is also a weaker variant of SETH known as the Exponential Time Hypothesis (ETH). This conjecture merely states that \ksat requires exponential time for all $k$ without making any statements about particular constants in the exponent. 

\begin{conjecture}[ETH \cite{IP99}]
    Every algorithm solving the Satisfiability of $3$-CNF formulas on $n$ variables must need $\Omega(2^{\delta n})$ time in expectation for some $\delta>0$. 
\end{conjecture}

\subsection{Lattice problems}
\label{sec:lattice-problems}

We give the formal definition of a lattice as well as the terms we use to work with lattices, such as the rank of a lattice and the distance of a lattice from a vector. 

\begin{definition}[Lattice, basis, coordinates, rank, and distance]
  \label{def:lattice}
  A \emph{lattice} $\mathcal{L} \subset \mathbb{R}^d$ is a set of all the linear combinations of linearly independent basis vectors $\bm{b}^{(1)}, \ldots, \bm{b}^{(n)} \in \mathbb{R}^d$ with integer coefficients,
  \begin{align*} 
    \mathcal{L}=\mathcal{L} (\bm{b}^{(1)}, \ldots, \bm{b}^{(n)}) \coloneqq
     \left\{ \sum_{i = 1}^n c_i  \bm{b}^{(i)} | c_i \in \mathbb{Z} \right\}. 
   \end{align*}
   Here $n$ is the \emph{rank} of the lattice $\mathcal{L}$ and $d$ is the \emph{ambient dimension} of $\mathcal{L}$.
  
   The \emph{basis} of a lattice can also be interpreted as a matrix $B \coloneqq \left(\begin{smallmatrix} | &  & |\\ \bm{b}^{(1)} & \cdots & \bm{b}^{(n)}\\ | &  & | \end{smallmatrix}\right)$ which relates every lattice point $\bm{v} \in \mathcal{L}$ to its coordinates $\bm{y} \in \mathbb{Z}^n$ by $\bm{v} = B \bm{y}$.  

   The \emph{distance} (with respect to the $\ell_p$-norm) between a vector $\bm{t}$ and a lattice $\mathcal{L}$, denoted by $\mathrm{dist}_p(\mathcal{L}, \bm{t})$, is defined as $\mathrm{dist}_p (\mathcal{L}, \bm{t}) = \min_{\bm{v} \in \mathcal{L}} \| \bm{v} - \bm{t} \|_p$. 
 \end{definition}

Next we present the formal definition of the closest vector problem.
\begin{definition}[The Closest Vector Problem]
  For $p \in [1, \infty)$ and $\gamma \geq 1$, the \emph{$\gamma$-approximate Closest Vector Problem} with respect to the $\ell_p$-norm, denoted \gammacvp, is defined as follows. Given a lattice $\mathcal{L}$ with $n$ (potentially sparse) basis vectors $\bm{b}^{(1)}, \ldots, \bm{b}^{(n)} \in \mathbb{R}^d$, a target vector $\bm{t} \in \mathbb{R}^d$, and a positive real number $r$, subject to the promise that either $\mathrm{dist}_p (\mathcal{L}, \bm{t}) \leq r$ or $\mathrm{dist}_p (\mathcal{L}, \bm{t}) > \gamma r$, decide if $\mathrm{dist}_p (\mathcal{L}, \bm{t}) \leq r$. 
\end{definition}

When $\gamma = 1$ we refer to the problem as \cvpsubp. Please note that an instance of the \gammacvp problem is fully defined by the triplet $(B, \bm{t}, r)$. 

Reductions to \cvpsubp from problems with binary variables like \ksat or \weightedmaxcut often create instances where the closest vector to the target always has binary coordinates. Since these instances are potentially easier than general instances, we define these instances as also belonging to the Binary Closest Vector Problem to better distinguish them.

\begin{definition}[Binary Closest Vector Problem]
    For $p \in [1, \infty)$, $\gamma \geq 1$, and distinct $a,b \in \mathbb{Z}$, the \emph{$\gamma$-approximate $\{a,b\}$-Binary Closest Vector Problem} with respect to the $\ell_p$-norm, which we denote by \gammabinarycvp[p], is a variant of the \gammacvp problem with an additional promise. The promise is as follows: for each instance $(B,\bm{t},r)$ of the problem with ambient dimension $d$, the coordinates $\bm{y}$ of the closest lattice vector $B \bm{y} \in \mathcal{L}$ to the target $\bm{t}$ are \emph{promised} to be in $\{a,b\}^d$. 
\end{definition}

Note that \gammabinarycvp is only defined for a certain subset of triplets $(B, \bm{t}, r)$ valid for \gammacvp. Since whether \gammabinarycvp captures the full power of \gammacvp is still an open problem, many reductions to \gammacvp are also reductions to \gammabinarycvp. For example, the reductions in \cite{BGS17} are to \gammabinarycvp with $\{a,b\}=\{0,1\}$\footnote{See for example, their discussion of their "identity gadget", as well as the last two paragraphs of the proof of Theorem 5.1.}, and the reduction in \cite{ABGS21} are to \gammabinarycvp with $\{a,b\}=\{-1,1\}$\footnote{See, for example, the introduction to section three ``We start by showing that it suﬃces to define isolating parallelepipeds over $\{-1,1\}$''.}. \gammabinarycvp is not to be confused with the Subset Closet Vector Problem defined below. 

\begin{definition}[Subset Closet Vector Problem]
    For $p \in [1, \infty)$ and $\gamma \geq 1$, \emph{$\gamma$-approximate Subset Closest Vector Problem} with respect to the $\ell_p$-norm, which we denote by $\gamma$-$(0,1)$-\cvpsubp, is a variant of the \gammacvp where only the distance from the target to the lattice points with coefficients in ${0,1}$ are considered. 
\end{definition}

\noindent Note any $(B, \bm{t}, r)$ triplets valid for \gammacvp is also valid for $\gamma$-$(0,1)$-\cvpsubp. 

We emphasize a crucial distinction between \gammabinarycvp and $\gamma$-$(0,1)$-\cvpsubp: \gammabinarycvp is strict subset of \gammacvp and thus lower bounds \gammacvp, while the time complexity of $\gamma$-$(0,1)$-\cvpsubp is currently incomparable to the time complexity of \gammacvp. In other words, every instance of \gammabinarycvp is also an instance of \gammacvp—an oracle for \gammacvp would return the same answer as an oracle for \gammabinarycvp on the on every $(B, \bm{t}, r)$ triple valid for \gammabinarycvp. An oracle for \gammacvp would return a \emph{different} answer than $\gamma$-$(0,1)$-\cvpsubp on certain lattices. When \gammabinarycvp has $\{a,b\}=\{0,1\}$, then it can also be solved by algorithms for $\gamma$-$(0,1)$-\cvpsubp.

\paragraph{Known Upper Bounds} 
The best known time upper bound for exact or near-exact \cvpsubp is $O(2^{n+o(n)})$ \cite{ADS15, AS18}. \cite{EV22} gives an $O(2^{0.802 n + \delta})$-time upper-bound for \gammacvp where $\gamma$ is some constant dependent on $\delta$. We do not know of any quantum algorithms with better time complexity in this regime. There are polynomial-time upper bounds for a subset of \gammacvp instances with $\gamma = O(2^{\sqrt{n}})$ \cite{EH22, ABCG22}. In this work we show some $O(2^{(1-\delta) n})$-time upper bounds for \gammabinarycvp[p] for some small positive values of $\delta$ and $p$ (see \cref{rmk:alg-for-binary-cvp-2}). See \cref{fig:fg-cvpvgamma} for a graph of known upper and lower bounds for \gammacvptwo. For a discussion of known classical lower bounds, see the introduction. 

\cite{KS20} showed a way to find upper bounds for $\gamma$-$(0,1)$-\cvpsubp and \gammabinarycvp for sufficiently large $\gamma$. This is discussed in detail in \cref{sec:faster-alg-max-cut}. After this paper was first publicized, \cite{AK25} showed an $O(2^{\omega n /3 + o(n)})$ time algorithm for $\gamma$-$(0,1)$-\cvpsubp dependent on the matrix multiplication constant $\omega$. This gives an $O(2^{0.79n})$ time algorithm at the time of their writeup. 


\subsection{\maxtwolintwo and Minimum Uncut}

\label{sec:max-cut}

The minimum uncut problem (\weightedmaxcut) and its generalization \maxtwolintwo are famous NP-complete problems; however, the complexity of their approximate variant is not settled due to their connection to the Unique Games Conjecture (see \cref{conj:ugc} and the discussion of the lower bound by \cite{KKMO07}). The approximate variant is the simplest \uniquegames instance with the minimum alphabet size of 2. Since many aspects of the complexity of the maximum cut problem are inherited from the \uniquegames problem, we will start by defining and discussing the complexity of the more general \uniquegames problem first, and then we will define and discuss the complexity of \weightedmaxcut and \maxtwolintwo. In this section we will often refer to \weightedmaxcut as \maxcutforreal (Maximum Cut) to emphasize its similarities with and relationship to \maxtwolintwo and \uniquegames. 

\subsubsection{Unique Games}
Unique Games is a particularly interesting category of constraint satisfaction problems (CSPs), where the goal of the problem is to find an assignment of values to a fixed set of variables that satisfies the maximum number of constraints on the values of the variables (see the survey \cite{Kho10} for more information). For our purposes, it is sufficient to identify Unique Games with its complete problem \uniquelabel, which we now formally define. 

\begin{definition}[Unique Label Cover Value]
    Define a \emph{Unique Label Cover} instance $\mathcal{G} := (V,W, E \subset V \times W, \Sigma, \pi_{(v,w) \in E}: \Sigma \to \Sigma)$ where $((V,W),E)$ is a bipartite, undirected graph, $\Sigma$ is a set of labels to be applied to the vertices of the graph, and each $\pi_{(v,w)}$ is a permutation of $\Sigma$. Then the value of the \uniquelabel instance is defined as 
    \begin{equation*}
        \mathrm{val} ( {\mathcal{G}} ) := \max_{f: V \cup W \to \Sigma} \left\{ \Pr_{(v,w) \in E}  \left[ \pi_{(v,w)} (f(v)) = f(w) \right] \right\}.
    \end{equation*}
\end{definition}

Here is an intuitive explanation of the problem. We are given a bipartite graph $(V,W),E$ where each vertex can be assigned a label from an alphabet $\Sigma$ and each edge is assigned a permutation on $\Sigma$; the vertices are the variables of our problem, the labels are values, and the permutations on the edges are the constraints. So if $(v,w)\in V \times W$ is an edge in our graph, then given any assignment to $v$ there is a unique assignment to $w$ which satisfies the permutation on $(v,w)$ and vice versa. We would like to find an assignment $f$ over the vertices $V$ and $W$ which maximizes the fraction of satisfied edges.

\begin{definition}[The Unique Label Problem]
    The \emph{$(1-\varepsilon, 1-\varsigma)$-\uniquelabel Problem} is the problem of deciding, given a Unique Label Cover instance $\mathcal{G}$ subject to the promise that either $\mathrm{val}(\mathcal{G}) > 1-\varepsilon$ or $\mathrm{val}(\mathcal{G}) < 1-\varsigma$, if  $\mathrm{val}(\mathcal{G}) > 1-\varepsilon$.
\end{definition}

\begin{remark}
    This problem is sometimes also referred to as the \uniquelabel problem with \emph{completeness} $1-\varepsilon$ and \emph{soundness} $1-\varsigma$ due to a connection to probabilistically checkable proofs. The distance between the completeness and the soundness is sometimes called the \emph{gap}; we can also call this a \emph{$(1-\varepsilon, 1-\varsigma)$-gap} Unique Game problem. 
\end{remark}

One can check if $\mathcal{G}$ has value 1 in linear time by the following method: since any assignment to a vertex enforces assignments on its neighbors, one can simply start from some $v \in V$, choose a random assignment from $\Sigma$ and check if it propagates into an assignment that satisfies all the constraints. However, it is not clear what happens to the complexity if the value is slightly less than 1. In particular, the following conjecture remains open. 

\begin{conjecture}[Unique Games Conjecture] \label{conj:ugc}
    For every constant $1 > \varepsilon > 0$, there exists $\Sigma$ with $|\Sigma| = O(1)$ such that it is NP-hard to distinguish between two instances of \uniquelabel with the following values:

    \emph{YES instance}: $\mathrm{val} ( {\mathcal{G}} ) > 1- \varepsilon$; 
    
    \emph{NO instance}: $\mathrm{val} ( {\mathcal{G}} ) < \varepsilon$.
    
\end{conjecture}

We refer the reader to \cite{Kho10} for the consequences and significance of the conjecture. Also using standard techniques, one can assume that the underlying graph is a regular graph (therefore the $V$ side is $d$-regular). Recently \cite{KMS23} resolved the case where the completeness is $1/2 - \varepsilon$ and soundness is $\varepsilon$; however, \cref{conj:ugc} still remains open.

\subsubsection{Maximum Cut and \maxtwolintwo}
Here, we will define unweighted \maxtwolintwo and its associated maximum cut problem. 
First, we define unweighted \maxtwolintwo, which is the Unique Games instance with alphabet size 2.

\begin{definition}[gap unweighted \maxtwolintwo]
    For $\varepsilon, c \in (0, 1)$, the $(1-\varepsilon, 1- \varepsilon^c)$-gap unweighted \maxtwolintwo problem is defined as follows. Given $n$ variables, $x_1, \ldots, x_n$, a set of linear constraints of the form $x_{e_1} \oplus x_{e_2} = c_{e}$ where $c_{e} \in \{0,1\}$ called $E$ (of size $m$) 
    and weight function $w:E \rightarrow \mathbb{R}^+$ (normalized so that $\sum_{e \in E} w(e) = 1$) such that $w(e)=1/|E|$ for all $e \in E$
    , decide if there exists a Boolean assignments to $x_i$'s such that:

    \emph{YES instance}: $\sum_{e \in E} w ( e ) \cdot \mathbbm{1}_{ x_{e_1} \oplus x_{e_2} = c_{e} }  > 1- \varepsilon$; 
    
    \emph{NO instance}: $\sum_{e \in E} w ( e ) \cdot \mathbbm{1}_{ x_{e_1} \oplus x_{e_2} = c_{e} }  < 1- \varepsilon^c$.
\end{definition}

Intuitively, we are given $m$ many linear constraints (over $\mathbb{F}_2$), each with two variables. The goal is to distinguish between the case where there exists an assignment which satisfies {\em most} of the constraints, versus the case where all assignments satisfies slightly less constraints. Note that using standard Gaussian Elimination, one can easily check if all constraints can be satisfied simultaneously. The problem instead asks if there exists an assignment which satisfies {\em most} of the constraints. Note that we also defined (weighted) \maxtwolintwo in the introduction. 

Now we proceed to defining the Maximum Cut problem, which is a special case or sub-instances of (weighted) \maxtwolintwo.

\begin{definition}[cut value, maximum cut ratio]
  The \emph{cut value} of a partition $P$ of a weighted graph $G = ([n], E)$ with weights given by the function $w:E \rightarrow \mathbb{R}^+$, denoted by $\mathrm{cutvalue} (G,w,P)$, is the total weight of cut the edges. Then the \emph{maximum cut ratio} of a weighted graph is $\mathrm{maxcutratio} (G,w) \coloneqq \max_P \frac{\mathrm{cutvalue} (G,w,P)}{w_\mathrm{tot}}$, where $w_\mathrm{tot}=\sum_{e \in E} w(e)$ represents the total weight of all edges. 
\end{definition}

\begin{definition}[The approximate Max-Cut problem]
    For $\varepsilon, c \in (0, 1)$, the \gapweightedmaxcut problem is defined as follows. Given the adjacency list $E \subseteq [n] \times [n]$ for a weighted graph $G = ([n], E)$ with $n$ vertices, $m := |E|$ edges, and weight function $w:E \rightarrow \mathbb{R}^+$ for which it is promised that either $\mathrm{maxcutratio} (G,w) \geq 1 - \varepsilon$ or $\mathrm{maxcutratio} (G,w) < 1 - \varepsilon^c$, decide if $\mathrm{maxcutratio} (G,w) \geq 1 - \varepsilon$. 
\end{definition}

\begin{remark}
Observe that the Maximum Cut problem is essentially {\bf a special case} of \maxtwolintwo, where we set all the constraints to be of the form $x_i \oplus x_j = 1$. Therefore any reduction from \maxtwolintwo or its gap variant {\bf implies} a reduction from Max-Cut or its gap variant. To phrase our result as general as possible, we show a reduction from gap \maxtwolintwo in \cref{sec:main-reduction}. The theorem then would also immediately imply a reduction from gap Max-Cut.
\end{remark}




Problems with smaller gaps are harder than problems with larger gaps. However there is a sense in which all $(1-\varepsilon, 1-\varepsilon^c)$-gap Unique Game instances with $c < 1/2$, including \maxcutforreal, are equally easy: using the parallel repetition technique~\cite{Raz98, Rao08, Hol09} (alas with a polynomial blow-up in the instance size), any algorithm that solves a $(1-\varepsilon', 1-\varepsilon'^c)$-gap Unique Game with $c \leq 1/2$ can solve any $(1-\varepsilon, 1-C'\varepsilon^{1/2})$-gap Unique Game for some specific $C'>0$. This ``parallel repetition'' regime is also where the previously known upper bounds on \gapweightedmaxcut apply. 

\paragraph{Known Landscape}

Since the known landscape for \maxcutforreal and \maxtwolintwo are equivalent from an algorithmic perspective, we only refer to \maxcutforreal throughout the remaining section to match with previous results.

Observe that \maxcutforreal can be solved by a naive brute force search over all possible assignments or partitions. Since each partition represents a binary choice for each and every vertex the classical naive search must try $O(2^n)$ many partitions in expectation and a quantum naive search \cite{Gro96} must try $O(2^{n/2})$ many partitions in expectation. Williams~\cite{Wil05} then shows an algorithm which finds the exact value of the cut (or total weights of satisfied constraints) in $1.8^n$-time.

But we are then interested in whether the approximation variant makes the problem any easier. Since the approximate variant of \maxcutforreal is a special case of \uniquegames, all \uniquegames upper bounds apply. These upper bounds are well-summarized in Khot's survey~\cite{Kho10}; here we will highlight the algorithms most relevant to our regime of interest. First we mention two polynomial-time algorithms: \cite{CMM06a} showed an algorithm which decides the $(1- \varepsilon, 1 - C \sqrt{ \varepsilon })$-gap problem for some sufficiently large constant $C$, and \cite{CMM06b} showed an algorithm which decides the $(1 - \varepsilon, 1 - \varepsilon  \sqrt{ \log n })$-gap problem. \cite{ABS15}'s breakthrough result showed a sub-exponential-time algorithm for deciding even smaller polynomial gaps between the YES and NO case; specifically, they showed a $2^{n^k}$-time algorithm which decides the  $(1 - \varepsilon, 1 - \sqrt{\varepsilon / k^3})$-gap problem for any $k$ larger than $\log \log n / \log n$ \cite{Ste10}. Their algorithm traces out a difficulty curve for Unique Games that increases as soundness approaches $1 - \sqrt{\varepsilon}$, where the parallel repetition regime ends  (see the solid line in \cref{fig:fg-maxcut}).

There is also a \maxcutforreal/\weightedmaxcut specific algorithm on the arXiv which solves $\alpha$-approximate \weightedmaxcut in $\exp(n/2^{\Omega(\alpha^2)})$-time \cite{MT18}\footnote{Note that a \weightedmaxcut algorithm is listed as an open question in the published version of this article and is present in the arXiv version which was submitted later.}. Here $\alpha = \varsigma/\varepsilon$ is the approximation ratio of a \esgap \weightedmaxcut instance. Note that even for fixed $c$, the ratio $\alpha$ for \gapweightedmaxcut still changes as a function of $\varepsilon$.

\begin{figure}[!ht]
  \centering
  \includegraphics[width=0.6\textwidth]{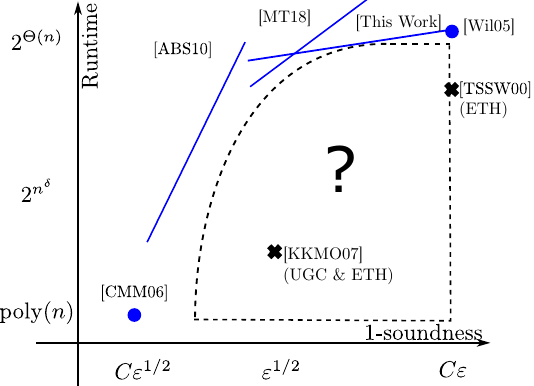}
  \caption{Known upper and lower bounds for \gapweightedmaxcut}
  \label{fig:fg-maxcut}
\end{figure}


\paragraph{Ideal lower bounds}
The holy-grail lower bound is to show NP-hardness (i.e. a polynomial sized reduction from 3SAT) for $(1 - \varepsilon, 1 - C \cdot \varepsilon^{1/2})$-gap \maxcutforreal for some sufficiently large constant $C$. Such a result combined with Parallel Repetition techniques~\cite{Raz98, Rao08, Hol09} would prove the Unique Games Conjecture. While no such result is known, it is known that the Unique Games Conjecture implies that $(1 - \varepsilon, 1 - C \cdot \varepsilon^{1/2})$-gap \maxcutforreal is NP-hard due to \cite{KKMO07}, which showed a surprising mapping from any $(1 - \varepsilon, 1 - C \cdot \varepsilon^{1/2})$-gap \maxcutforreal instance to an easier \uniquelabel instance; this result gives a $2^{n^\delta}$-time lower bound for some $\delta>0$ assuming the Unique Games Conjecture and Exponential Time Hypothesis. Currently, the best known lower bound is \cite{TSSW00} which showed that it is NP-hard to distinguish instances of \maxcutforreal with cut ratio $1- \varepsilon$ and cut ratio $ 1 - C' \varepsilon$ (using a linear sized reduction for some small $C' > 1$).

\subsection{Approximate Nearest Neighbor} 

Approximate Nearest Neighbor is a well-studied search problem in the context of sub-linear time algorithms and hashing algorithms. Given a set $A$ of $N$ points, Nearest Neighbor is the problem of finding the exact closest point in $A$ to each query point $\bm{x}$ received in an online fashion; this requires linear time for each query. However, sub-linear query times are possible through the use of preprocessing if the queries are relaxed to accept any point in $A$ close to $\bm{x}$ instead of only accepting the exact closest point. Approximate Nearest Neighbor is the variant of Nearest Neighbor with these more relaxed queries. We give a formal definition of the decision variant of Approximate Nearest Neighbor below:

\begin{definition} [Approximate Nearest Neighbor]
    The \emph{$\gamma$-Approximate Nearest Neighbor} problem (\emph{\gammaANN}) with ambient dimension $d$ in the $\ell_p$-norm for approximation factor $\gamma > 1$ is defined as follows. First, preprocess $N$ vectors $A \subset \mathbb{R}^d$ into some data structure given a parameter $r \in \mathbb{R}^+$. Then for each query vector $\bm{x} \in \mathbb{R}^d$ decide if $\min_{\bm{a} \in A} ||\bm{a} - \bm{x}||_p \leq r$ using the data structure. Every query vector $\bm{x}$ is subject to the promise that either $\min_{\bm{a} \in A} ||\bm{a} - \bm{x}||_p \leq r$ or $\min_{\bm{a} \in A} ||\bm{a} - \bm{x}||_p \geq \gamma r$. 
\end{definition}

Note that \gammacvp can be thought of as a variant of \gammaANN with only one query to a very structured set of vectors. A \gammaANN instance is fully defined by the set $A$, the promise distance $r$, and the set of query vectors $\{\bm{x}^{(1)},\dotsc\}$. See the survey \cite{AIR18} for more information about the Approximate Nearest Neighbor problem. 

The following near optimal state-of-the-art algorithm is known due to \cite{AR15}:
\begin{restatable}[{\cite[Theorem 1.1]{AR15}}]{theorem}{aralgorithm}
    \label{thm:ar15-alg}
    For any fixed $\gamma > 1$, there is an algorithm which solves the \gammaANN problem with ambient dimension $d$ in the $\ell_2$ norm on $N$ points using $O(d \cdot N^{\rho + o(1)})$ query time, $O(N^{1+\rho+o(1)} + d \cdot N)$ space, and $O(d \cdot N^{1+\rho+o(1)})$ preprocessing time, where $\rho = \frac{1}{2 \gamma^2 - 1}$. There is also an algorithm for the $\ell_1$ norm with the same parameters except $\rho = \frac{1}{2 \gamma - 1}$. 
\end{restatable}

\section{Main reduction}
\label{sec:main-reduction}
In this section we will prove our main theorem, which we restate below. We will first present our reduction and then use it to prove the theorem. 

\restatemaintheorem*

\subsection{Reduction}
\label{red:main-reduction}
We first give closed form definitions of the basis $B$ and target vector $\bm{t}$ produced by our reduction in terms of the \maxtwolintwo instance it receives. For the clarity of the closed form definitions, we represent the \maxtwolintwo instance as the triple $(n, (E^{(0)},E^{(1)}), w)$. $n \in \mathbb{N}$ is the number of variables in the instance. $E^{(0)},E^{(1)}$ are ordered sets of constraints representing the equality and inequality constraints respectively. We denote the set of all constraints as $E := E^{(0)} \cup E^{(1)}$ and the number of constraints as $m_0, m_1, m := |E^{(0)}|, |E^{(1)}|, |E|$. Each constraint is an unordered pair of variables which we write as $(i,j) \in E$ with $i<j$. $w:E \rightarrow [1,\infty)$ is a 1-normalized weight function\footnote{Our results can still be applied to general weight functions because any weight function can be normalized through scaling the weights by a factor of $w_\mathrm{min}^{-1}$.} with $w(e)=1$ for at least one $e \in E$. For convenience, we denote the minimum weight of a weight function as $w_\mathrm{min} \coloneqq \min_{e \in E} w(e)$ and the total weight of a weight function as $w_\mathrm{tot} \coloneqq \sum_{e \in E} w(e)$. 

We can think of $\bm{t}$ as a series of 2-row blocks with each of the first $m_1$ blocks dedicated to an inequality constraint and each of the remaining $m_0$ blocks dedicated to an equality constraint, as follows:

\[
\bm{t} \coloneqq \sum_{k = 1}^{m_1}  (0 \cdot \bm{e}_{2 k - 1} + \iota w'(E^{(1)}_k) \cdot \bm{e}_{2 k}) + 
    \sum_{k = m_1 + 1}^{m_1+m_0}  (0 \cdot \bm{e}_{2 k - 1} + 1 \cdot \bm{e}_{2 k}),
\]
where $w'(e) \coloneqq \sqrt[p]{w(e)}$ is an $\ell_p$-norm adjusted weight function, $\bm{e}_1, \ldots, \bm{e}_{2 m}$ is an orthonormal basis for $\mathbb{R}^{2m}$, and $\iota > 1$ is a parameter. 

Next we define each of the $n$ columns of $B$, one for each variable. Again we can think of it as a stack of 2-row blocks, one for each constraint. For convenience we think of $E$ in a similar way, with $E_{1:m_1}=E^{(1)}$ and $E_{m_1+1:m}=E^{(0)}$. Since the column $B_{\cdot, i}$ is a basis vector representing the variable $v_i$, it only has nonzero entries in the rows for constraints that $v_i$ participates in. We give a slightly different construction when $i$ is the second variable in a constraint to ensure that $B$ is linearly independent. 

\begin{align*}
B_{\cdot, i} \coloneqq &
    \sum_{k \in [m_1] |E_k = (i, \cdot)}  (-\bm{e}_{2 k - 1} + \iota w'(E_k) \cdot \bm{e}_{2 k}) + 
    \sum_{k \in [m_1]|E_k = (\cdot, i)}  (\bm{e}_{2 k - 1} + \iota w'(E_k) \cdot \bm{e}_{2k}) + 
    \\ &
    \sum_{k \in \{(m_1+1)..m\}|E_k = (i, \cdot)} (-\iota w'(E_k) \cdot\bm{e}_{2 k - 1} + \bm{e}_{2 k}) + 
    \sum_{k \in \{(m_1+1)..m\}|E_k = (\cdot, i)}  (\iota w'(E_k) \cdot\bm{e}_{2 k - 1} + \bm{e}_{2 k}) 
\end{align*}

For \minusesgap \maxtwolintwo instances, we set $r \coloneqq  (m (1 - \varepsilon) + \iota^p \varepsilon w_\mathrm{tot})^{1 / p}$. When $\iota \gg m$ is sufficiently large we claim that this gives the output instances an approximation factor of $\gamma = \left( \frac{m (1 - \varsigma) + \iota^p \varsigma w_\mathrm{tot}}{m (1 - \varepsilon) + \iota^p \varepsilon w_\mathrm{tot}} \right)^{1 / p}$. We can choose $\iota$ to be as large as we want as long as it is still small enough to compute on. For example, if we wanted $\iota$ to be representable in $O(\log n)$ bits we could choose $\iota := 2^{32} n^{32} \gg n^2 \geq m$.

\subsection{Proof of main theorem}
\begin{proof}[Proof of \cref{thm:weighted-main-theorem}]
    We claim that the reduction in \cref{red:main-reduction} is a linear-time quantum and classical reduction from \minusesgap \maxtwolintwo to \esgamma-\cvpsubp that uses one basis vector per variable as required. Since the reduction uses one basis vector per variable and is classical by definition, it suffices to show the following: that the reduction can be performed in linear time, that the reduction is valid quantum reduction, that the reduction always produces a valid \cvpsubp instance, and that the reduction produces instances which satisfy the approximation factor promise, also known as the completeness and soundness guarantee. 

    For some suitable sparse vector representation, it is clear that the reduction can be performed in linear ($O(m)$) time by adding a constant number of entries to the basis vectors for each constraint ($\bm{t}$ and $r$ can also be constructed by their definitions in linear time). 

    To call this reduction a quantum reduction we would ideally want to be able to treat it like a perfect quantum oracle which can take a coherent superpostion of valid input instances and return a coherent superposition of output instances where output instances that are the same constructively interfere with each other. It suffices to show that the reduction can be performed on a quantum computer, that the reduction always halts after predictable amount of time, and that the same output instances are always represented the same way. Since the reduction is deterministic and does not delete information, it is also reversible (meaning that it's possible to compute the exact representation of the input \minusesgap \maxtwolintwo instance from the output \gammacvp instance) and thus can be performed on a quantum computer. Since we've already defined an input representation that is unique and our reduction is deterministic, there is only one possible output for each input and for our purposes different inputs are supposed to have different outputs. Because this is a Karp reduction (meaning it has exactly one input and one output), this suffices to show that the same output instances always have the same representation. Since our reduction clearly halts, always makes progress, and we already established it runs in linear time, it always halts after predictable amount of time. 

    To show that $(B,\bm{t},r)$ is a valid \cvpsubp instance, it suffices to show that the columns of $B$ are linearly independent.  Consider an arbitrary lattice basis vector $B_{\cdot, i}$. Without loss of generality we can assume that the $i^\text{th}$ variable participates in a constraint\footnote{Otherwise we can safely drop the $i^\text{th}$ variable from the input instance.}, so $B_{\cdot, i}$ has a non-zero entry. Let $k$ be an index such that $i \in E_k$ and let $j$ be the other variable in $E_k$. Then $B_{\cdot, i}$ and $B_{\cdot, j}$ are the only lattice basis vectors with non-zero coefficients for $\bm{e}_{2 k - 1}$ and $\bm{e}_{2 k}$, so $B_{\cdot, i}$ is linearly independent from $\{ B_{\cdot, 1}, \ldots, B_{\cdot, n} \} - \{ B_{\cdot, i}, B_{\cdot, j} \}$.  Furthermore, $B_{\cdot, i}$ and $B_{\cdot, j}$ must be linearly independent because their subvectors in the $\bm{e}_{2 k - 1}$ and $\bm{e}_{2 k}$ dimensions are linearly independent. Thus we have established that $B_{\cdot, 1}, \ldots, B_{\cdot, n}$ are linearly independent. 

    In the remaining part of the proof, we will give upper and lower bounds on the distance of the lattice to the target when the reduction is given good and bad instances, then use those bounds to show that the instances satisfy the required approximation factor promise. To do so, we will assume for now that that the closest lattice point to the target {\bf always has binary coefficients} due to the structure of our lattice and the target $\bm{t}$. This assumption is to be proved in \cref{lemma:binary-coord-lemma}.

    For any assignment $v_1,\dots,v_n$ to the variables we can construct a lattice point with binary coefficients $\bm{u} = \sum_{i\in[n]}v_i B_{\cdot, i}$, and for any lattice point with binary coefficients $\bm{u} = \sum_{i\in[n]}v_i B_{\cdot, i}$ we can construct an assignment $v_1,\dots,v_n$ to the variables. The distance of such a lattice point to the target is given by 
    \begin{align*}
        \mathrm{dist}_p(\bm{u}, \bm{t})^p &= \sum_{k \in [m]} \mathrm{dist}_p(\bm{u}_{2k-1:2k}, \bm{t}_{2k-1:2k})^p 
        = \sum_{k \in E, (i,j) = E_k} \mathrm{dist}_p(v_i B_{\cdot, i} + v_i B_{\cdot, j}, \bm{t}_{2k-1:2k})^p
        \\ &= \sum_{k \in [m]} \begin{cases}
            1^p & k^{\mathrm{th}}~  \text{constraint is satisfied}\\
            \iota^p w'(E_k)^p = \iota^p w(E_k) & k^{\mathrm{th}}~  \text{constraint is unsatisfied}
        \end{cases}
        \\ &= (\text{\# of satisfied constraints}) + \iota^p(\text{total weight of unsatisfied constraints}).
    \end{align*}
    So when the reduction is given an instance of \maxtwolintwo with a variable assignment that satisfies a $(1-\varepsilon)$-fraction of the constraints there is a lattice point which has distance $(m (1 - \varepsilon) + \iota^p \varepsilon w_\mathrm{tot})^{1 / p}$ from the target. This gives an upper bound on the distance from the lattice to the target. 

    When the reduction is given an instance of \maxtwolintwo where every variable assignment satisfies less than a $(1-\varsigma)$-fraction then every lattice point with 0-1 coefficients is more than distance $(m (1 - \varsigma) + \iota^p \varsigma w_\mathrm{tot})^{1 / p}$ away from the target. By \cref{lemma:binary-coord-lemma} one of those points is the closest point in the lattice to the target, so this gives a lower bound on the distance to the target for every lattice point. 

    The above upper bound and lower bound give the promise gap (between $r$ and $\gamma r$) between the close and far \cvpsubp instances produced by the reduction. All that remains is to check that they give the required approximation factor. 
    \[
    \gamma = \frac{\gamma r}{r} = \frac{(m (1 - \varsigma) + \iota^p \varsigma w_\mathrm{tot})^{1 / p}}{(m (1 - \varepsilon) + \iota^p \varepsilon w_\mathrm{tot})^{1 / p}} \approx \left(\frac{\iota^p \varsigma w_{\mathrm{tot}}}{\iota^p \varepsilon w_{\mathrm{tot}}} \right)^{1/p} = \sqrt[p]{\varsigma/\varepsilon}
    \]
    as required, when $\iota^p \gg m$. 
\end{proof}

\begin{lemma}
    \label{lemma:binary-coord-lemma}
    The closest lattice point to the target $\bm{t}$ always has 0-1 coefficients on the basis vectors in the \cvpsubp instances produced by reduction in \cref{red:main-reduction}. 
\end{lemma}
\begin{proof}
  It suffices to show that for any lattice point with non-binary coordinates, there is a lattice point with binary coordinates that is at least as close to $\bm{t}$. Pick an arbitrary coordinate $\bm{y} \in \mathbb{Z}^n$ with at least one non-binary entry and let $\widetilde{\bm{y}} \in \{ 0, 1 \}^n$ denote the binarization of $\bm{y}$ given by 
  $\widetilde{\bm{y}}_i =
  \begin{cases}
    0 & \bm{y}_i \leq 0\\
    1 & \bm{y}_i \geq 1
  \end{cases}$; 
  we will show that $\| B \widetilde{\bm{y}} - \bm{t} \|_p < \| B \bm{y} - \bm{t} \|_p$, which implies the previous statement. Pick an arbitrary $k \in [m]$, and compare the distance between $B \widetilde{\bm{y}}$ and $\bm{t}$ to the distance between $B \bm{y}$ and $\bm{t}$ along the $\bm{e}_{2 k - 1}$ and $\bm{e}_{2 k}$ bases. Let $(i, j) = E_k$ denote the $k^{\mathrm{th}}$ constraint, and let $V(i) = \widetilde{\bm{y}}_i$ be the assignment to the variables according to $\widetilde{\bm{y}}$.
  
  If $E_k$ is an inequality constraint and $V$ satisfies it then $\{ \widetilde{\bm{y}}_i, \widetilde{\bm{y}}_j \} = \{ 0, 1 \}$. Without loss of generality let $\widetilde{\bm{y}}_i = 0$; then $\bm{y}_i \leq 0$ and $\bm{y}_j \geq 1$ so
  \begin{align*}
    | (B \widetilde{\bm{y}} - \bm{t})_{2 k - 1} | & =  | -
    \widetilde{\bm{y}}_i + \widetilde{\bm{y}}_j - 0 | = 1\\
    & \leq  \\
    | - \bm{y}_i + \bm{y}_j - 0 | & =  | ( \bm{y}_i
    B_{\cdot, i} + \bm{y}_j B_{\cdot, j} - \bm{t})_{2 k - 1} | = |
    (B \bm{y} - \bm{t})_{2 k - 1} |
  \end{align*}
  and
  \begin{align*}
    |  (B \widetilde{\bm{y}} - \bm{t})_{2 k} | = | \iota w'(E_k) \cdot (
    \widetilde{\bm{y}}_i + \widetilde{\bm{y}}_j - 1) | & =  0\\
    & \leq  | (B \bm{y} - \bm{t})_{2 k} |
  \end{align*}
  so $B \widetilde{\bm{y}}$ is at least as close to $\bm{t}$ as $B \bm{y}$ along the $\bm{e}_{2 k - 1}$ and $\bm{e}_{2 k}$ bases. Note that the first inequality is strict when at least one of $\bm{y}_i, \bm{y}_j$ is non-binary.
  
  If $E_k$ is an inequality constraint and $V$ does not satisfy it then $\widetilde{\bm{y}}_i$ and $\widetilde{\bm{y}}_j$ are either both 0 or both 1, so
  \begin{align*} 
    | (B \widetilde{\bm{y}} - \bm{t})_{2 k - 1} | = 0 \leq |
     (B \bm{y} - \bm{t})_{2 k - 1} | 
   \end{align*}
  and $\bm{y}_i$ and $\bm{y}_j$ are either both $\leq 0$ or both $\geq 1$ so
  \begin{align*} 
    |  (B \widetilde{\bm{y}} - \bm{t})_{2 k} | = \iota w'(E_k) \leq
     | \iota w'(E_k) \cdot (\bm{y}_i + \bm{y}_j - 1) | = | (B \bm{y}
     - \bm{t})_{2 k} | 
   \end{align*}
  so $B \widetilde{\bm{y}}$ is at least as close as $B \bm{y}$ to $\bm{t}$ along the $\bm{e}_{2 k - 1}$ and $\bm{e}_{2 k}$ bases. Note that the second inequality is strict when at least one of $\bm{y}_i, \bm{y}_j$ is non-binary.

  If $E_k$ is an equality constraint and $V$ satisfies it then $\widetilde{\bm{y}}_i$ and $\widetilde{\bm{y}}_j$ are either both 0 or both 1, so
  \begin{align*} 
    | (B \widetilde{\bm{y}} - \bm{t})_{2 k - 1} | = 0 \leq |
     (B \bm{y} - \bm{t})_{2 k - 1} | 
   \end{align*}
  and $\bm{y}_i$ and $\bm{y}_j$ are either both $\leq 0$ or both $\geq 1$ so
  \begin{align*} 
    |  (B \widetilde{\bm{y}} - \bm{t})_{2 k} | = 1 \leq
     | \bm{y}_i + \bm{y}_j - 1| = | (B \bm{y}
     - \bm{t})_{2 k} | 
   \end{align*}
  so $B \widetilde{\bm{y}}$ is at least as close as $B \bm{y}$ to $\bm{t}$ along the $\bm{e}_{2 k - 1}$ and $\bm{e}_{2 k}$ bases. Note that the second inequality is strict when at least one of $\bm{y}_i, \bm{y}_j$ is non-binary.

  If $E_k$ is an equality constraint and $V$ does not satisfy it then $\{ \widetilde{\bm{y}}_i, \widetilde{\bm{y}}_j \} = \{ 0, 1 \}$. Without loss of generality let $\widetilde{\bm{y}}_i = 0$; then $\bm{y}_i \leq 0$ and $\bm{y}_j \geq 1$ so
  \begin{align*}
    | (B \widetilde{\bm{y}} - \bm{t})_{2 k - 1} | &= \iota w'(E_k)\\
    & \leq 
    \\ | -\iota w'(E_k) \bm{y}_i + \iota w'(E_k) \bm{y}_j - 0 | 
    & =  | ( \bm{y}_iB_{\cdot, i} + \bm{y}_j B_{\cdot, j} - \bm{t})_{2 k - 1} | 
    = |(B \bm{y} - \bm{t})_{2 k - 1} |
  \end{align*}
  and
  \begin{align*}
    |  (B \widetilde{\bm{y}} - \bm{t})_{2 k} | = |\widetilde{\bm{y}}_i + \widetilde{\bm{y}}_j - 1| & =  0 \leq  | (B \bm{y} - \bm{t})_{2 k} |
  \end{align*}
  so $B \widetilde{\bm{y}}$ is at least as close to $\bm{t}$ as $B \bm{y}$ along the $\bm{e}_{2 k - 1}$ and $\bm{e}_{2 k}$ bases. Note that the first inequality is strict when at least one of $\bm{y}_i, \bm{y}_j$ is non-binary.
  
  Since $k$ was chosen arbitrarily, this shows that $\widetilde{\bm{y}}$ is always at least as close as $\bm{y}$ to $\bm{t}$ along each computational basis; and since in both cases there is a strict inequality when at least one of $\bm{y}_i, \bm{y}_j$ is non-binary, there is at least one computational basis where $B \widetilde{\bm{y}}$ is closer to $\bm{t}$ than $B \bm{y}$ is. Thus $\| B \widetilde{\bm{y}} - \bm{t} \|_p < \|  B \bm{y} - \bm{t} \|_p $.  
\end{proof}

\section{Fine-grained hardness of approximate CVP}
\label{sec:fg-hardness-of-cvp}

In this section we quantify which advancements in upper bounds for \gammacvp would result in interesting new algorithms for \maxtwolintwo due to our reduction. 


\label{sec:classical-faster-cvp-faster-mc}

\paragraph{Classical CVP hardness and approximation factor}

We can give the following corollary directly from \cref{cor:equal-time} and the state of \minusesgap \maxtwolintwo and \weightedmaxcut algorithms:

\begin{corollary}
    Any sub-exponential-time classical algorithm for $O(1)$-CVP$_p$ (with $p\in [1, \infty)$) would also be an algorithm for \gapweightedmaxcut and \eecgap \maxtwolintwo that bests all known algorithms, including \cite{ABS15}, \cite{CMM06b}, \cite{Wil05}, \cite{MT18}, and the one in \cref{sec:faster-alg-max-cut}, in running time, applicable approximation regime, or both.

    Furthermore, any $O(2^{n/2})$-time classical algorithm for $c_1$-CVP$_p$ (with $p\in [1, \infty)$ and $c_1$ a constant dependent on $p$ and unspecified large constants in \cite{MT18}) would also be an algorithm for \gapweightedmaxcut and \eecgap \maxtwolintwo that bests all known algorithms, including \cite{ABS15}, \cite{CMM06b}, \cite{Wil05}, \cite{MT18}, and the one in \cref{sec:faster-alg-max-cut}, in running time, applicable approximation regime, or both.
\end{corollary}

Since our reduction extends across a large range of approximation parameters, in \cref{tab:multiple-classical-bounds} we give two different regimes for \minusesgap \maxtwolintwo, the maximum $\gamma$ given by our reduction from instances in those regimes, and the values $\varepsilon$ and $\varsigma$ such that a \minusesgap \maxtwolintwo is in the listed regime and our reduction produces a \gammacvp instance. Remember from \cref{sec:max-cut} that \minusesgap \maxtwolintwo is NP-hard for constant $\varepsilon$ and $\varsigma < c_h \varepsilon$ for some constant $c_h$ \cite{TSSW00} and NP-hard under the Unique Games Conjecture for constant $\varepsilon$ and $\varsigma \leqslant \sqrt{\varepsilon}$ \cite{KKMO07}.

\begin{table}[h]
    \centering
    \begin{tabular}{|c|c|c|c|c|}
    \hline
        \minusesgap \maxtwolintwo regime & Maximum $\gamma$ & $\varepsilon$ &  $\varsigma$ \\
        \hline
        NP-hard & $\sqrt[p]{c_h}$ & $O(1)$ & $c_h \varepsilon$ \\
        NP-hard (under UGC) & Any Constant & $\gamma^{-p/2}$ & $\sqrt{\varepsilon}$ \\
        \hline
    \end{tabular}
    \caption{By \cref{cor:equal-time}, a $2^{\Omega(n)}$-time lower bound on the regime on the first column implies a $2^{\Omega(n)}$-time lower-bound on \gammacvp with the $\gamma$ given by the second column. The third and fourth columns give a choice of values for $\varepsilon$ and $\varsigma$ such that \minusesgap \maxtwolintwo is in the listed hardness regime and our reduction reduces \minusesgap \maxtwolintwo instances to \gammacvp instances with the list value of $\gamma$.}
    \label{tab:multiple-classical-bounds}
\end{table}

\paragraph{Quantum CVP hardness}
Since \cref{thm:weighted-main-theorem} applies to quantum algorithms, the analysis above also applies to quantum algorithms as well. An $O(\exp(n^{1-\delta}))$-time quantum algorithm with $\delta > 0$ for \gammabinarycvp implies an $O(\exp(n^{1-\delta}))$-time quantum algorithm with $\delta > 0$ for \gapweightedmaxcut. 

\begin{corollary}
    Any sub-exponential-time quantum algorithm for $O(1)$-CVP$_p$ (with $p\in [1, \infty)$ would also be an algorithm for \gapweightedmaxcut and \eecgap \maxtwolintwo that bests all known classical algorithms as well as the first quantum algorithm described in \cref{sec:faster-alg-max-cut}, in running time, applicable approximation regime, or both. 
    
    Futhermore, any $O(2^{n/3})$-time quantum algorithm for $c_1$-CVP$_p$ (with $p\in [1, \infty)$ and $c_1$ a constant dependent on $p$ and unspecified large constants in \cite{MT18}) would also be an algorithm for \gapweightedmaxcut and \eecgap \maxtwolintwo that bests all known classical algorithms as well as the first quantum algorithm described in \cref{sec:faster-alg-max-cut}, in running time, applicable approximation regime, or both. 
\end{corollary}

This provides some evidence against weak-exponential or sub-exponential quantum algorithms for \gammabinarycvp for $\gamma = O (1)$ since such an algorithm would be a huge breakthrough. This also provides some evidence against a super-polynomial quantum advantage for \gammacvp since \gapweightedmaxcut is not thought to be amenable to super-polynomial quantum advantage. 

\section{Barriers against Conditional Lower Bounds for Max-Cut}

\subsection{Barriers against SETH-based optimality of Max-Cut} \label{sec:SETH_barrier}

\subsubsection{Previous results}

A series of works culminating in in \cite{DV14} and \cite{Dru15} investigated the compression of Boolean Satisfiability problems, yielding results ruling out one-sided probabilistic Turing reductions and non-adaptive probabilistic Karp reductions from $k$-SAT to any decision problem with instance size $O(n^{k-\delta})$ for $\delta > 0$ unless the polynomial time hierarchy collapses. If a problem $P$ has a compression of fixed polynomial size then these results rule out a polynomial-time fine-grained reduction from $k_0$-SAT to $P$ for some $k_0>0$ since such a reduction, combined with the compression of $P$, would be a compression for $k_0$-SAT of sub-$O(n^{k_0})$ size. These results also provide substantial barriers against proving the SETH-hardness of $P$ since the main technique for invoking SETH requires proving the the existence of a fine-grained polynomial-time reduction of \ksat to $P$ for every $k$ and a compression is any reduction from one decision problem to another decision problem with a smaller input size. 

\cite{AK23} used this strategy to show barriers against proving the SETH-hardness of \cvpsubp for positive even $p$ by providing compression reductions of such \cvpsubp instances to instances of $\mathrm{CVP}_p^{\mathrm{IP}}$ or $\mathrm{CVP}_p^{\mathrm{mvp}}$ (CVP-type problems purpose-built for compression). Here it is worth emphasizing that for the purposes of invoking the theorems of \cite{DV14} and \cite{Dru15} it suffices to know that instances of \cvpsubp with positive even $p$ can be reduced to instances of some decision problems with input sizes that are fixed polynomials in $n$; neither the definition nor the difficulty of these problems matter to them. Below we reproduce two of the theorems from \cite{AK23} showing barriers against the SETH-hardness of \cvptwo:

\begin{theorem}
  [{\cite[Theorem 7.1]{AK23}}]
  \label{thm:AK23-7.1}
  For any constants $a, a_1 > 0$, there exists a constant $k_0$ such that for any $k > k_0$ there does not exist a polynomial time probabilistic reduction with no false negatives from $k$-SAT on $n$ variables to $\mathrm{CVP}_2$ on rank $n^{a_1}$ lattices that makes at most $n^a$ calls to the $\mathrm{CVP}_2$ oracle unless $\mathrm{coNP} \subseteq \mathrm{NP} / \mathrm{Poly}$. 
\end{theorem}

This theorem results from combining their compression of $\mathrm{CVP}_2$ with a result from \cite{DV14}; the next one results from combining the same compression with a result from \cite{Dru15}. Note that the existence of non-uniform, statistical zero-knowledge proofs for all languages in $\mathrm{NP}$ is a stronger condition than $\mathrm{NP} \subseteq \mathrm{NP}/\mathrm{Poly} \cap \mathrm{coNP}/\mathrm{Poly}$.

\begin{theorem}
  [{\cite[Theorem 7.7]{AK23}}]
  \label{thm:AK23-7.7}
  For any constants $a, a_1 > 0$, there exists a constant $k_0$ such that for any $k > k_0$ there does not exist a (non-adaptive) randomized polynomial time reduction from $k$-SAT on $n$ variables to $\mathrm{CVP}_2$ on rank $n^{a_1}$ lattices with constant error bound which makes at most $n^a$ calls to the $\mathrm{CVP}_2$ oracle unless there are non-uniform, statistical zero-knowledge proofs for all languages in $\mathrm{NP}$. 
\end{theorem}

\subsubsection{Our corollaries}


The main technique for applying the lower-bound from SETH or its variants to a problem is to provide a fine-grained reduction from \ksat to that problem. An interesting corollary of our reduction is that most fine-grained reductions from \ksat to \gapweightedmaxcut are unattainable as they would violate \cref{thm:AK23-7.1} or \cref{thm:AK23-7.7}, as illustrated in \cref{fig:fg-ksat-reduction-map}. 

\begin{figure}[!ht]
  \centering
  \includegraphics[width=0.8\textwidth]{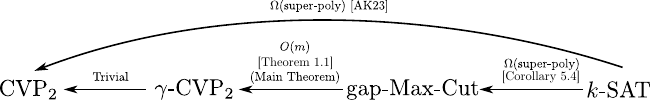}
  \caption{Map of reductions from $k$-SAT to CVP$_2$ with edges labeled by reduction sizes.}
  \label{fig:fg-ksat-reduction-map}
\end{figure}

Our reduction from \esgap \weightedmaxcut to \gammacvptwo is also a polynomial-time, polynomial-sized, and deterministic reduction from exact \weightedmaxcut to \cvptwo since \gammacvptwo is a strictly easier variant of \cvptwo; so our reduction also has the effect of applying the no-go conditions from the theorems of \cite{AK23} to \weightedmaxcut, yielding two results. The first is against reductions with one-sided error from $k$-$\mathrm{SAT}$ to \weightedmaxcut:

\begin{corollary}
  \label{thm:classical-mc-no-turing}
  For any constants $a, a_1 > 0$, there exists a constant $k_0$ such that for any $k > k_0$ there does not exist a polynomial time probabilistic reduction with no false negatives from \ksat on $n$ variables to \weightedmaxcut on $n^{a_1}$ vertices that makes at most $n^a$ calls to the \weightedmaxcut oracle unless $\mathrm{coNP} \subseteq \mathrm{NP} / \mathrm{Poly}$. 
\end{corollary}
\begin{proof}
  Suppose that for some constants $a,a_1>0$ there exists a reduction for any $k \in \mathbb{N}$ from \ksat on $n$ variables to \weightedmaxcut on $n^{a_1}$ vertices that makes at most $n^a$ calls to the \weightedmaxcut oracle. Then by \cref{thm:weighted-main-theorem} there is a polynomial time probabilistic reduction with no false negatives from \ksat on $n$ variables to \gammacvptwo on rank $O(n^{a_1})$ lattices that makes at most $n^a$ calls to the \gammacvptwo oracle. This implies that $\mathrm{coNP} \subseteq \mathrm{NP} / \mathrm{Poly}$ by \cref{thm:AK23-7.1}. 
\end{proof}

With an analogous argument using \cref{thm:AK23-7.7} instead of \cref{thm:AK23-7.1} we can make the following statement against Karp reductions (with two-sided error):
\begin{corollary}
  \label{thm:no-randomized-reductions-from-seth}
  For any constants $a , a_1 > 0$, there exists a constant $k_0$ such that for any $k > k_0$ there does not exist a (non-adaptive) randomized polynomial time reduction from \ksat on $n$ variables to \weightedmaxcut on $n^{a_1}$ vertices with constant error bound which makes at most $n^a$ calls to the \weightedmaxcut oracle unless there are non-uniform, statistical zero-knowledge proofs for all languages in $\mathrm{NP}$. 
\end{corollary}
\begin{proof}
    Follows from \cref{thm:weighted-main-theorem} and \cref{thm:AK23-7.7} in the same way as the previous theorem. 
\end{proof}

These theorems show a substantial barrier against proving the SETH-hardness of \weightedmaxcut since, assuming that the polynomial-time hierarchy does not collapse, they suggest that there are no fine-grained reductions from \ksat to \weightedmaxcut. 
In other words, these imply that one cannot hope to show $2^{\delta n}$-time lower bound against \weightedmaxcut for {\bf any constant} $\delta > 0$ assuming SETH. Note that, for the same reasons as with \cvptwo and \gammacvptwo, these theorems apply to all variations of unweighted/weighted and gap/exact Max-Cut since their reductions to (exact, weighted) \weightedmaxcut are trivial. 

\subsection{Barriers against QSETH-based optimality of Max-Cut}
\label{sec:qseth_barrier}


We can give a quantum version of \cref{thm:no-randomized-reductions-from-seth} by quantizing \cref{thm:AK23-7.7} using the compression in \cite{AK23} for $\mathrm{CVP}_2$ together with the quantum incompressibility theorem in \cite{Dru15}. We assume that the reader is familiar with quantum computing. 

\subsubsection{Preliminaries}

We reproduce the compression for \cvptwo below:  
\begin{theorem}
  [{\cite[Theorem 6.2]{AK23}}]
  \label{thm:AK23-6.2}
  For any positive integers $m,n$, given a $\mathrm{CVP}_2(B,\bm{t},r)$ instance on bithlength $\eta$ where $B \in \mathbb{Z}^{m \times n}$ is a basis of a lattice $\mathcal{L}$, target $\bm{t} \in \mathbb{Z}^m$ and $r>0$, there exists a $poly(n, m, \eta)$ time algorithm that reduces it to a $2^{n^2}$-$\mathrm{CVP}^{\mathrm{IP}}$ instance of size $O(n^8)$ bits. 
\end{theorem}

Now we reproduce the quantum incompressibility theorem and its corresponding definition of quantum compressibility. 

\begin{definition}
    [{\cite[Definition 6.2]{Dru15}}]
    \footnote{We have presented this definition with what we perceive to be typographical errors corrected.}
    Let $L$ be a language, and let $f:\{0,1\}^* \rightarrow \{0,1\}$ be a Boolean function. Let $t_1(n),t_2(n):\mathbb{N}^+ \rightarrow \mathbb{N}^+$ and $\xi(n):\mathbb{N}^+ \rightarrow [0,1]$ be given, with $t_1,t_2$ polynomially bounded. A \emph{quantum $f$-compression reduction for $L$} with parameters $t_1(n),t_2(n),\xi(n)$ is a mapping $R(x^1, \dotsc, x^{t_1(n)})$ outputting a mixed state $\rho$, along with a family of (not necessarily efficiently preformable) binary quantum measurements $\{\mathcal{M}_n\}_{n>0}$ on $t_2(n)$-qubit states. We require the following properties: for all $(x^1, \dotsc, x^{t_1(n)}) \in \{0,1\}^{t_1(n) \times n}$,
    \begin{enumerate}
        \item the state $\rho = R(x^1, \dotsc, x^{t_1(n)})$ is on $t_2(n)$ qubits;
        \item we have $\Pr[\mathcal{M}_m(\rho)=f \left( L(x^1), \dotsc, L(x^{t_1(n)})\right)] \geq 1 - \xi(n)$.
    \end{enumerate}
    If $R$ as above is computable in quantum polynomial time, we say that $L$ is {\normalfont QPT-$f$-compressible} with parameters $(t_1(n),t_2(n),\xi(n))$. 
\end{definition}

\begin{theorem}[{\cite[Theorem 6.19]{Dru15}}]
   \footnote{We have presented this theorem with what we perceive to be typographical errors corrected.}
  \label{thm:Dru15-6.19}
  Let $L$ be any language. Suppose there is a QPT-$\mathrm{OR}$-compression reduction $R(x^1,...,x^{t_1(n)}) : \{0,1\}^{t_1(n)} \times n \rightarrow \mathrm{MS}_{t_2(n)}$ for $L$ with polynomially bounded parameters $t_1(n), t_2(n) : \mathbb{N}^+ \rightarrow \mathbb{N}^+$, and error bound $\xi(n) < 0.5$. Let $\hat{\delta} := \min \{ \sqrt{\frac{\ln 2}{2} \cdot \frac{t_2(n)}{t_1(n)}}, 1-2^{-\frac{t_2(n)}{t_1(n)}-2} \}$. 
  
  1. If for some $c > 0$ we have $(1-2 \xi(n))-\hat{\delta} \geq \frac{1}{n^c}$, then there is a nonuniform (classical, deterministic) polynomial-time many-to-one reduction from $L$ to a problem in $\mathrm{pr}$-$\mathrm{QIP[2]}$. 
  
  2. If we have the stronger bound $(1-2 \xi(n))^2-\hat{\delta} \geq \frac{1}{n^c}$, then L has a nonuniform (classical, deterministic) polynomial-time many-to-one reduction to a problem in $\mathrm{pr}$-$\mathrm{QSZK}$.
\end{theorem}

We also need to use the reduction in \cite{DV14} from \orlang{\threesat} to \ksat. 

\begin{lemma}
    [\cite{DV14}]
    \label{lemma:DVM14-or3sat-to-ksat}
    For any integer $k \geq 3$, there is a polynomial-time reduction from \orlang{\threesat} to \ksat that maps $t$ tuples of 3-CNF formulas, each of bitlength $n$, to a $k$-CNF formula on $O\big(n \cdot \max\big(n,t^{1/k+o(1)}\big)\big)$ variables. 
\end{lemma}

\begin{proof}
    Combine \cite[Lemma 3]{DV14} and \cite[Lemma 6]{DV14} with the trivial reduction of $k$-Clique to $k$-Vertex Cover. 
\end{proof}

\subsubsection{Our results}

Now we are ready to give the quantized version of \cref{thm:AK23-7.7} (\cite[Theorem 7.7]{AK23}), following a similar proof. 

\begin{theorem}
    \label{thm:quantum-cvp2-incompressible}
    For any constants $a, a_1 > 0$, there exists a constant $k_0$ such that for any $k > k_0$ there does not exist a (non-adaptive) quantum polynomial-time reduction from \ksat on $n$ variables to \cvptwo on rank $n^{a_1}$ lattices with constant error bound which makes at most $n^a$ calls to the \cvptwo oracle unless there are non-uniform, quantum statistical zero-knowledge proofs for all languages in $\mathrm{NP}$. 
\end{theorem}

\begin{proof}
    Let $c_2:=8 c_1$ and $k_0=c+c_2$. Assume that there exists a non-adaptive quantum polynomial time reduction $R$ from \ksat on $n$ variables to \cvptwo on rank $n^{c_1}$ lattices with constant error bound, and that the reduction makes at most $n^c$ calls to the \cvptwo oracle. By \cref{thm:AK23-6.2}, a \cvptwo instance with a rank $n^{c_1}$ lattice can be compressed to an instance of size $\widetilde{O}(n^{c_2})$. So we can achieve quantum instance compression of \ksat of size $\widetilde{O}(n^{c+8 c_1})=O(n^{k_0})$ by implementing the \cvptwo oracle in $R$ as a combination of the reduction in \cref{thm:AK23-6.2} and an oracle for \twontwocvpip. 
    
    By \cref{lemma:DVM14-or3sat-to-ksat}, we also get quantum instance compression of \orlang{\threesat} to $O(n \cdot \max\big(n,t^{1/k+o(1)}\big)$ bits. When $t$ is a sufficiently large polynomial in $n$, the compressed instance has bitlength $O(t^{\delta + o(1)}) = O(t \log t)$ for some $0<\delta<1$ since $k_0$ is a constant less than $k$. Now we can apply \cref{thm:Dru15-6.19} with $t_1=t$ and $t_2=O(t \log t)$ since our quantum instance compression of \orlang{\threesat} can also be converted into a QPT-$\mathrm{OR}$-compression of \threesat by returning the requests to the \twontwocvpip oracle and dropping the parts of the reduction that use the results from the oracle. This conversion is possible because the reduction is non-adaptive, and it works because any algorithm for \twontwocvpip (no matter how inefficient) combined with the parts of the reduction that use the results of the oracle calls is a binary quantum measurement which decides \orlang{\threesat} with constant error. Applying \cref{thm:Dru15-6.19} gives us that there is a non-uniform, classical, deterministic, polynomial-time many-to-one reduction from \threesat to a problem in $\mathrm{pr}$-$\mathrm{QSZK}$. This implies that there are non-uniform quantum statistical zero-knowledge proofs for all languages in $\mathrm{NP}$. 
\end{proof}

Now that we have a quantized version of \cref{thm:AK23-7.7}, we can finally give our quantized version of \cref{thm:no-randomized-reductions-from-seth}. 

\begin{corollary}
    \label{thm:quantum-no-ksat-to-ug-reduction}
    For any constants $a, a_1 > 0$, there exists a constant $k_0$ such that for any $k > k_0$ there does not exist a (non-adaptive) quantum polynomial time reduction from \ksat on $n$ variables to \weightedmaxcut on $n^{a_1}$ vertices with constant error bound which makes at most $n^a$ calls to the \weightedmaxcut oracle unless there are non-uniform, statistical quantum zero-knowledge proofs for all languages in $\mathrm{NP}$.
\end{corollary}
\begin{proof}
    By \cref{thm:weighted-main-theorem} and \cref{thm:quantum-cvp2-incompressible}.
\end{proof}

This corollary shows a substantial barrier against proving the QSETH-hardness of \weightedmaxcut since, assuming that there is a language in $\mathrm{NP}$ without a statistical quantum zero-knowledge proof, there are no quantum fine-grained reductions from \ksat to \weightedmaxcut. Although the complexity of $\mathrm{QSZK}$ is not as well understood as its classical counterpart $\mathrm{SZK}$ (where it is known that the polynomial hierarchy collapses if $\mathrm{NP}$ is contained in it), there is an oracle relative to which $\mathrm{QSZK}$ is not in $\mathrm{UP}$ \cite{MW18} ($\mathrm{UP}$ is a subset of $\mathrm{NP}$\footnote{See \url{https://complexityzoo.net/Complexity_Zoo:U\#up} for the relationship between UP and NP.}). Again, note that these theorems apply to all variants of unweighted/weighted gap/exact \weightedmaxcut.



\section{Faster algorithms for \minusesgap \maxtwolintwo}
\label{sec:faster-alg-max-cut}

\subsection{Preliminaries}
Restricting a \cvpsubp instance to a binary choice of coefficients reduces the total number of possible lattice vectors to exactly $2^n$; \cite{KS20} discovered a reduction from \binarycvp to $\mathrm{Nearest}$ $\mathrm{Neighbor}$ that exploits the structure of \binarycvp to further restrict the search space. The main idea is as follows: let $0<a<1$ be a fraction; if the original lattice $\mathcal{L}$ with $n$ bases vectors was split into two lattices $\mathcal{L}_1$ and $\mathcal{L}_2$, where $\mathcal{L}_1$ had $a n$ of the basis vectors of $\mathcal{L}$ and $\mathcal{L}_2$ had the remaining $(1 - a) n$ basis vectors, then the closest vector $\bm{v} \in \mathcal{L}$ to $\bm{t}$ in the original lattice is the sum of two vectors from the split lattices $\bm{v}_1 \in \mathcal{L}_1$ and $\bm{v}_2 \in \mathcal{L}_2$; so $\bm{v}_2 - \bm{t}$ must be close to $\bm{v}_1$. By taking all the binary vectors of $\mathcal{L}_1$ to be the points in a $\mathrm{Nearest}$ $\mathrm{Neighbor}$ instance, one only needs to make $2^{(1-a) n}$ queries to $\mathrm{Nearest}$ $\mathrm{Neighbor}$ (one for each binary lattice vector in $\mathcal{L}_2$) to find a pair of vectors $\bm{v}_1 + \bm{v}_2$ close to $\bm{t}$ since the closest vector in the original lattice is promised to be binary. With a sufficiently fast algorithm for $\mathrm{Nearest}$ $\mathrm{Neighbor}$ this could yield an algorithm which solves \binarycvp in sub-$2^n$-time. 

\begin{theorem}
  [{\cite[Section 3]{KS20}}]
  \label{thm:ks20-reduction}
    For any $\rho > 0$, $C > 1$, and $0 < a < 1$, if there is an algorithm for \gammaANN with ambient dimension $d$ under the $\ell_p$ norm with $N^C$ preprocessing time and $N^{\rho}$ query time, then there is an  $O(2^{C a n} + 2^{(1 + \rho) (1-a) n})$-time algorithm for \gammabinarycvp with ambient dimension $d$, where $n$ denotes the rank of the input lattice. 
\end{theorem}

Now we use the following algorithm for \gammaANN which is fast enough to solve \gammabinarycvp[2] with $\gamma > 1$ in sub-$2^n$ time. 

\aralgorithm*

Note that \cite{AR15} propose algorithms for the search variant of \gammaANN which can be trivially adapted to solve the decision variant.

\begin{remark}
    \label{rmk:alg-for-binary-cvp-2}
    By combining the previous two theorems and picking the parameter $a=\frac{1}{2}$, we know that there is an $O(2^{ \left( \frac{1}{2} + \frac{1}{4 \gamma^2 - 2} + o(1) \right) n})$-time algorithm for \gammabinarycvp[2] and an $O(2^{ \left( \frac{1}{2} + \frac{1}{4 \gamma - 2} + o(1) \right) n})$-time algorithm for \gammabinarycvp[1]. 
\end{remark}

\subsection{An $O(2^{n(\frac{1}{2} + \frac{\varepsilon}{4\varsigma} + o(1))})$-time classical algorithm for \minusesgap \maxtwolintwo}

Since our reduction from $(1-\varepsilon, 1-\varsigma)$-gap \maxtwolintwo to \gammabinarycvp[2] is linear in time and size, this naturally results in an algorithm of the same time complexity for $(1-\varepsilon, 1-\varsigma)$-gap \maxtwolintwo. 

\begin{theorem}
    \label{thm:classical-max-cut-alg}
    For any $\varepsilon,\varsigma \in (0,1)$, there is an $O(2^{n(\frac{1}{2} + \frac{1}{4\varsigma/\varepsilon - 2} + o(1))})$-time algorithm for $(1-\varepsilon, 1-\varsigma)$-gap \maxtwolintwo. 
\end{theorem}
\begin{proof}
    By \cref{thm:weighted-main-theorem} we can reduce any $(1-\varepsilon, 1-\varsigma)$-gap \maxtwolintwo instance to a \gammabinarycvp[2] instance which can then be reduced to a \gammaANN instance and solved in the specified time by \cref{rmk:alg-for-binary-cvp-2}. 
\end{proof}

\begin{remark}
    If $\varsigma / \varepsilon$ is an increasing function of $n$, then the above algorithm has time complexity $O(2^{ \left( \frac{1}{2} + o(1) \right) n})$. 
\end{remark}

\begin{remark}
    When $\varsigma / \varepsilon$ is greater than a sufficiently large constant, then the above algorithm has time complexity close to $O(2^{ \left( \frac{1}{2} + \frac{\varepsilon}{4\varsigma} + o(1) \right) n})$. 
\end{remark}

\subsection{An $O(2^{n(\frac{1}{3} + \frac{\varepsilon}{6\varsigma} + o(1))})$-time quantum algorithm for \minusesgap \maxtwolintwo}
In this section, we present a Grover-based quantum algorithm for gap Max-Cut that provides a polynomial speedup over the classical algorithm presented above. This quantum algorithm assumes that we can efficiently load data from a quantum memory, i.e., that the operation
\begin{align*}
  \ket{i, b, z_1,\ldots, z_m} \mapsto \ket{i, z_i, z_1,\ldots, z_{i-1},b, z_{i+1},\ldots, z_m}
\end{align*}
takes $O(\log m)$ time, where $\ket{z_1, \ldots, z_m}$ is some data in a quantum random access memory with $m$ qubits. This operation is needed when we use Grover's algorithm to query the data structure in superposition.

\begin{theorem}
    \label{rmk:qalg-for-binary-cvp-2}
    There exists an $O\Bigl(2^{ \bigl( \frac{1}{3} + \frac{2}{6 \gamma^2 - 3} + o(1) \bigr) n}\Bigr)$-time quantum algorithm for \gammabinarycvp[2]. 
\end{theorem}
\begin{proof}
  This quantum algorithm is based on using Grover's algorithm~\cite{Gro96} to speed up the search process of the reduction in \cref{thm:ks20-reduction}, while still using \cref{thm:ar15-alg} as the underlying data structure. More specifically, instead of partitioning the basis vectors into halves, we choose $a=1/3$ in \cref{thm:ks20-reduction}. Then we create a data structure for the lattice with $2^{n/3}$ binary points according to the preprocessing procedure in \cref{thm:ar15-alg}. We then query the data structure for each of the $2^{2n/3}$ binary points in the other lattice using the query procedure in \cref{thm:ar15-alg}, and it is straightforward to use Grover's algorithm to reduce the number of queries to $O(2^{n/3})$. Considering the time for each query and the time for constructing the data structure, the theorem follows.
\end{proof}

\begin{remark}
    When $\varsigma / \varepsilon$ is greater than a sufficiently large constant, then the above algorithm has time complexity close to $O(2^{ \left( \frac{1}{2} + \frac{\varepsilon}{6\varsigma} + o(1) \right) n})$. 
\end{remark}

Note that in the above theorem, when $\varsigma / \varepsilon$ is an increasing function of $n$, the dependence on $\gamma$ in the time complexity of the gap \maxtwolintwo will be hidden in $o(1)$. In addition, an intermediate of the reduction is a quantum algorithm for \gammabinarycvp[2], as noted in the following remark.

\begin{theorem}
    \label{thm:quantum-max-cut-alg}
    For any $\varepsilon,\varsigma \in (0,1)$, there is an $O(2^{n(\frac{1}{3} + o(\frac{1}{6\varsigma/\varepsilon - 3}) + o(1))})$-time quantum algorithm for $(1-\varepsilon, 1-\varsigma)$-gap \maxtwolintwo. 
\end{theorem}
\begin{proof}
    Apply \cref{thm:weighted-main-theorem} with $p=2$ and then use \cref{rmk:qalg-for-binary-cvp-2} to solve the resulting \binarycvp[2] instance.
\end{proof}

\bibliographystyle{alpha}
\bibliography{resources/cvpMaxCut}

@inproceedings{ABB+23,
  title = {Lattice {{Problems}} beyond {{Polynomial Time}}},
  booktitle = {Proceedings of the 55th {{Annual ACM Symposium}} on {{Theory}} of {{Computing}}},
  author = {Aggarwal, Divesh and Bennett, Huck and Brakerski, Zvika and Golovnev, Alexander and Kumar, Rajendra and Li, Zeyong and Peters, Spencer and {Stephens-Davidowitz}, Noah and Vaikuntanathan, Vinod},
  year = 2023,
  month = jun,
  pages = {1516--1526},
  publisher = {ACM},
  address = {Orlando FL USA},
  doi = {10.1145/3564246.3585227},
  urldate = {2024-08-31},
  isbn = {978-1-4503-9913-5},
  langid = {english}
}

@misc{ABCG22,
  title = {Quantum and {{Classical Algorithms}} for {{Bounded Distance Decoding}}},
  author = {Allen, Richard and Berker, Ratip Emin and Casacuberta, S{\'i}lvia and Gul, Michael},
  year = 2022,
  month = feb,
  urldate = {2024-10-03},
  langid = {english}
}

@inproceedings{ABGS21,
  title = {Fine-Grained Hardness of {{CVP}}({{P}}): Everything That We Can Prove (and Nothing Else)},
  shorttitle = {Fine-Grained Hardness of {{CVP}}({{P}})},
  booktitle = {Proceedings of the {{Thirty-Second Annual ACM-SIAM Symposium}} on {{Discrete Algorithms}}},
  author = {Aggarwal, Divesh and Bennett, Huck and Golovnev, Alexander and {Stephens-Davidowitz}, Noah},
  year = 2021,
  month = mar,
  series = {{{SODA}} '21},
  pages = {1816--1835},
  publisher = {{Society for Industrial and Applied Mathematics}},
  address = {USA},
  urldate = {2024-07-18},
  isbn = {978-1-61197-646-5}
}

@article{ABS15,
  title = {Subexponential {{Algorithms}} for {{Unique Games}} and {{Related Problems}}},
  author = {Arora, Sanjeev and Barak, Boaz and Steurer, David},
  year = 2015,
  month = nov,
  journal = {Journal of the ACM},
  volume = {62},
  number = {5},
  pages = {1--25},
  issn = {0004-5411, 1557-735X},
  doi = {10.1145/2775105},
  urldate = {2024-07-18},
  langid = {english}
}

@article{ABSS97,
  title = {The {{Hardness}} of {{Approximate Optima}} in {{Lattices}}, {{Codes}}, and {{Systems}} of {{Linear Equations}}},
  author = {Arora, Sanjeev and Babai, L{\'a}szl{\'o} and Stern, Jacques and Sweedyk, Z},
  year = 1997,
  month = apr,
  journal = {Journal of Computer and System Sciences},
  volume = {54},
  number = {2},
  pages = {317--331},
  issn = {0022-0000},
  doi = {10.1006/jcss.1997.1472},
  urldate = {2025-03-10}
}

@inproceedings{ACL+20,
  title = {On the Quantum Complexity of Closest Pair and Related Problems},
  booktitle = {35th Computational Complexity Conference ({{CCC}} 2020)},
  author = {Aaronson, Scott and Chia, Nai-Hui and Lin, Han-Hsuan and Wang, Chunhao and Zhang, Ruizhe},
  editor = {Saraf, Shubhangi},
  year = 2020,
  series = {Leibniz International Proceedings in Informatics (Lipics)},
  volume = {169},
  pages = {16:1--16:43},
  publisher = {Schloss Dagstuhl -- Leibniz-Zentrum f\"ur Informatik},
  address = {Dagstuhl, Germany},
  issn = {1868-8969},
  doi = {10.4230/LIPIcs.CCC.2020.16},
  isbn = {978-3-95977-156-6},
  urn = {urn:nbn:de:0030-drops-125681}
}

@inproceedings{ADS15,
  title = {Solving the {{Closest Vector Problem}} in 2\textasciicircum n {{Time}} -- {{The Discrete Gaussian Strikes Again}}!},
  booktitle = {2015 {{IEEE}} 56th {{Annual Symposium}} on {{Foundations}} of {{Computer Science}}},
  author = {Aggarwal, Divesh and Dadush, Daniel and {Stephens-Davidowitz}, Noah},
  year = 2015,
  month = oct,
  pages = {563--582},
  issn = {0272-5428},
  doi = {10.1109/FOCS.2015.41},
  urldate = {2024-10-03}
}

@misc{ADW+24,
  title = {More {{Asymmetry Yields Faster Matrix Multiplication}}},
  author = {Alman, Josh and Duan, Ran and Williams, Virginia Vassilevska and Xu, Yinzhan and Xu, Zixuan and Zhou, Renfei},
  year = 2024,
  month = oct,
  number = {arXiv:2404.16349},
  eprint = {2404.16349},
  primaryclass = {cs},
  publisher = {arXiv},
  doi = {10.48550/arXiv.2404.16349},
  urldate = {2025-03-28},
  archiveprefix = {arXiv}
}

@misc{AIR18,
  title = {Approximate {{Nearest Neighbor Search}} in {{High Dimensions}}},
  author = {Andoni, Alexandr and Indyk, Piotr and Razenshteyn, Ilya},
  year = 2018,
  month = jun,
  journal = {arXiv.org},
  urldate = {2024-09-26},
  howpublished = {https://arxiv.org/abs/1806.09823v1},
  langid = {english}
}

@inproceedings{AK23,
  title = {Why We Couldn't Prove {{SETH}} Hardness of the {{Closest Vector Problem}} for Even Norms!},
  booktitle = {2023 {{IEEE}} 64th {{Annual Symposium}} on {{Foundations}} of {{Computer Science}} ({{FOCS}})},
  author = {Aggarwal, Divesh and Kumar, Rajendra},
  year = 2023,
  month = nov,
  pages = {2213--2230},
  publisher = {IEEE},
  address = {Santa Cruz, CA, USA},
  doi = {10.1109/FOCS57990.2023.00138},
  urldate = {2024-07-18},
  copyright = {https://doi.org/10.15223/policy-029},
  isbn = {979-8-3503-1894-4}
}

@misc{AK25,
  title = {On {{Beating}} \$2\textasciicircum n\$ for the {{Closest Vector Problem}}},
  author = {Abboud, Amir and Kumar, Rajendra},
  year = 2025,
  month = jan,
  number = {arXiv:2501.03688},
  eprint = {2501.03688},
  primaryclass = {cs},
  publisher = {arXiv},
  doi = {10.48550/arXiv.2501.03688},
  urldate = {2025-03-28},
  archiveprefix = {arXiv}
}

@article{AR05,
  title = {Lattice Problems in {{NP}} Intersect {{coNP}}},
  author = {Aharonov, Dorit and Regev, Oded},
  year = 2005,
  month = sep,
  journal = {Journal of the ACM},
  volume = {52},
  number = {5},
  pages = {749--765},
  issn = {0004-5411},
  doi = {10.1145/1089023.1089025},
  urldate = {2020-08-31}
}

@inproceedings{AR15,
  title = {Optimal {{Data-Dependent Hashing}} for {{Approximate Near Neighbors}}},
  booktitle = {Proceedings of the Forty-Seventh Annual {{ACM}} Symposium on {{Theory}} of {{Computing}}},
  author = {Andoni, Alexandr and Razenshteyn, Ilya},
  year = 2015,
  month = jun,
  series = {{{STOC}} '15},
  pages = {793--801},
  publisher = {Association for Computing Machinery},
  address = {New York, NY, USA},
  doi = {10.1145/2746539.2746553},
  urldate = {2024-09-27},
  isbn = {978-1-4503-3536-2}
}

@article{AS18,
  title = {Just {{Take}} the {{Average}}! {{An Embarrassingly Simple}} 2\textasciicircum n-{{Time Algorithm}} for {{SVP}} (and {{CVP}})},
  author = {Aggarwal, Divesh and {Stephens-Davidowitz}, Noah},
  year = 2018,
  journal = {OASIcs, Volume 61, SOSA 2018},
  volume = {61},
  pages = {12:1-12:19},
  issn = {2190-6807},
  doi = {10.4230/OASICS.SOSA.2018.12},
  urldate = {2024-09-20},
  collaborator = {Seidel, Raimund},
  copyright = {Creative Commons Attribution 3.0 Unported license, info:eu-repo/semantics/openAccess},
  langid = {english}
}

@inproceedings{BGS17,
  title = {On the {{Quantitative Hardness}} of {{CVP}}},
  booktitle = {2017 {{IEEE}} 58th {{Annual Symposium}} on {{Foundations}} of {{Computer Science}} ({{FOCS}})},
  author = {Bennett, Huck and Golovnev, Alexander and {Stephens-Davidowitz}, Noah},
  year = 2017,
  month = oct,
  pages = {13--24},
  issn = {0272-5428},
  doi = {10.1109/FOCS.2017.11}
}

@misc{BPS19,
  title = {The {{Quantum Strong Exponential-Time Hypothesis}}},
  author = {Buhrman, Harry and Patro, Subhasree and Speelman, Florian},
  year = 2019,
  month = nov,
  journal = {arXiv.org},
  urldate = {2024-10-02},
  howpublished = {https://arxiv.org/abs/1911.05686v2},
  langid = {english}
}

@inproceedings{CMM06a,
  title = {Near-Optimal Algorithms for Unique Games},
  booktitle = {Proceedings of the Thirty-Eighth Annual {{ACM}} Symposium on {{Theory}} of {{Computing}}},
  author = {Charikar, Moses and Makarychev, Konstantin and Makarychev, Yury},
  year = 2006,
  month = may,
  series = {{{STOC}} '06},
  pages = {205--214},
  publisher = {Association for Computing Machinery},
  address = {New York, NY, USA},
  doi = {10.1145/1132516.1132547},
  urldate = {2024-07-23},
  isbn = {978-1-59593-134-4}
}

@inproceedings{CMM06b,
  title = {How to {{Play Unique Games Using Embeddings}}},
  booktitle = {Proceedings of the 47th {{Annual IEEE Symposium}} on {{Foundations}} of {{Computer Science}}},
  author = {Chlamtac, Eden and Makarychev, Konstantin and Makarychev, Yury},
  year = 2006,
  series = {{{FOCS}} '06},
  pages = {687--696},
  publisher = {IEEE Computer Society},
  address = {USA},
  doi = {10.1109/FOCS.2006.36},
  isbn = {0-7695-2720-5}
}

@article{Din02,
  title = {Approximating {{SVP-infinity}} to within Almost-Polynomial Factors Is {{NP-hard}}},
  author = {Dinur, Irit},
  year = 2002,
  month = aug,
  journal = {Theoretical Computer Science},
  series = {Algorithms and {{Complexity}}},
  volume = {285},
  number = {1},
  pages = {55--71},
  issn = {0304-3975},
  doi = {10.1016/S0304-3975(01)00290-0},
  urldate = {2020-08-31},
  langid = {english}
}

@article{DKRS03,
  title = {Approximating {{CVP}} to {{Within Almost-Polynomial Factors}} Is {{NP-Hard}}},
  author = {Dinur, I. and Kindler, G. and Raz, R. and Safra, S.},
  year = 2003,
  month = apr,
  journal = {Combinatorica},
  volume = {23},
  number = {2},
  pages = {205--243},
  issn = {0209-9683, 1439-6912},
  doi = {10.1007/s00493-003-0019-y},
  urldate = {2024-08-30},
  copyright = {http://www.springer.com/tdm}
}

@article{Dru15,
  title = {New {{Limits}} to {{Classical}} and {{Quantum Instance Compression}}},
  author = {Drucker, Andrew},
  year = 2015,
  month = jan,
  journal = {SIAM Journal on Computing},
  volume = {44},
  number = {5},
  pages = {1443--1479},
  publisher = {{Society for Industrial and Applied Mathematics}},
  issn = {0097-5397},
  doi = {10.1137/130927115},
  urldate = {2024-08-18}
}

@article{DV14,
  title = {Satisfiability {{Allows No Nontrivial Sparsification}} Unless the {{Polynomial-Time Hierarchy Collapses}}},
  author = {Dell, Holger and Van Melkebeek, Dieter},
  year = 2014,
  month = jul,
  journal = {Journal of the ACM},
  volume = {61},
  number = {4},
  pages = {1--27},
  issn = {0004-5411, 1557-735X},
  doi = {10.1145/2629620},
  urldate = {2024-08-01},
  langid = {english}
}

@misc{EH22,
  title = {An Efficient Quantum Algorithm for Lattice Problems Achieving Subexponential Approximation Factor},
  author = {Eldar, Lior and Hallgren, Sean},
  year = 2022,
  month = jan,
  urldate = {2024-10-03},
  langid = {english}
}

@book{Emd81,
  title = {{Another NP-complete partition problem and the complexity of computing short vectors in a lattice}},
  author = {van {Emde-Boas}, P.},
  year = 1981,
  publisher = {Department, Univ.},
  langid = {dutch}
}

@article{EV22,
  title = {Approximate {{CVPp}} in Time 20.802n},
  author = {Eisenbrand, Friedrich and Venzin, Moritz},
  year = 2022,
  month = mar,
  journal = {Journal of Computer and System Sciences},
  volume = {124},
  pages = {129--139},
  issn = {0022-0000},
  doi = {10.1016/j.jcss.2021.09.006},
  urldate = {2024-10-03}
}

@article{GJS76,
  title = {Some Simplified {{{\emph{NP}}}}-Complete Graph Problems},
  author = {Garey, M. R. and Johnson, D. S. and Stockmeyer, L.},
  year = 1976,
  month = feb,
  journal = {Theoretical Computer Science},
  volume = {1},
  number = {3},
  pages = {237--267},
  issn = {0304-3975},
  doi = {10.1016/0304-3975(76)90059-1},
  urldate = {2025-07-10}
}

@article{GMSS99,
  title = {Approximating Shortest Lattice Vectors Is Not Harder than Approximating Closest Lattice Vectors},
  author = {Goldreich, O. and Micciancio, D. and Safra, S. and Seifert, J. -P.},
  year = 1999,
  month = jul,
  journal = {Information Processing Letters},
  volume = {71},
  number = {2},
  pages = {55--61},
  issn = {0020-0190},
  doi = {10.1016/S0020-0190(99)00083-6},
  urldate = {2024-10-29}
}

@inproceedings{Gro96,
  title = {A Fast Quantum Mechanical Algorithm for Database Search},
  booktitle = {Proceedings of the 28th {{ACM}} Symposium on Theory of Computing ({{STOC}} 1996)},
  author = {Grover, Lov K},
  year = 1996,
  publisher = {ACM Press},
  doi = {10.1145/237814.237866}
}

@article{Hol09,
  title = {Parallel {{Repetition}}: {{Simplification}} and the {{No-Signaling Case}}},
  author = {Holenstein, Thomas},
  year = 2009,
  journal = {Theory of Computing},
  volume = {5},
  number = {1},
  pages = {141--172},
  issn = {1557-2862},
  doi = {10.4086/toc.2009.v005a008},
  urldate = {2024-09-10},
  langid = {english}
}

@inproceedings{HR07,
  title = {Tensor-Based Hardness of the Shortest Vector Problem to within Almost Polynomial Factors},
  booktitle = {Proceedings of the 39th {{Annual ACM Symposium}} on {{Theory}} of {{Computing}} ({{STOC}})},
  author = {Haviv, Ishay and Regev, Oded},
  year = 2007,
  month = jun,
  series = {{{STOC}} '07},
  pages = {469--477},
  address = {New York, NY, USA},
  doi = {10.1145/1250790.1250859},
  urldate = {2020-08-31},
  isbn = {978-1-59593-631-8}
}

@inproceedings{IP99,
  title = {The {{Complexity}} of K-{{SAT}}},
  booktitle = {Proceedings of the {{Fourteenth Annual IEEE Conference}} on {{Computational Complexity}}},
  author = {Impagliazzo, Russell and Paturi, Ramamohan},
  year = 1999,
  month = may,
  series = {{{COCO}} '99},
  pages = {237},
  publisher = {IEEE Computer Society},
  address = {USA},
  urldate = {2020-10-12},
  isbn = {978-0-7695-0075-1}
}

@incollection{Kar72,
  title = {Reducibility among {{Combinatorial Problems}}},
  booktitle = {Complexity of {{Computer Computations}}: {{Proceedings}} of a Symposium on the {{Complexity}} of {{Computer Computations}}, Held {{March}} 20--22, 1972, at the {{IBM Thomas J}}. {{Watson Research Center}}, {{Yorktown Heights}}, {{New York}}, and Sponsored by the {{Office}} of {{Naval Research}}, {{Mathematics Program}}, {{IBM World Trade Corporation}}, and the {{IBM Research Mathematical Sciences Department}}},
  author = {Karp, Richard M.},
  editor = {Miller, Raymond E. and Thatcher, James W. and Bohlinger, Jean D.},
  year = 1972,
  pages = {85--103},
  publisher = {Springer US},
  address = {Boston, MA},
  doi = {10.1007/978-1-4684-2001-2_9},
  urldate = {2025-07-15},
  isbn = {978-1-4684-2001-2},
  langid = {english}
}

@inproceedings{Kho10,
  title = {On the {{Unique Games Conjecture}} ({{Invited Survey}})},
  booktitle = {2010 {{IEEE}} 25th {{Annual Conference}} on {{Computational Complexity}}},
  author = {Khot, Subhash},
  year = 2010,
  month = jun,
  pages = {99--121},
  issn = {1093-0159},
  doi = {10.1109/CCC.2010.19},
  urldate = {2024-10-28}
}

@article{KKMO07,
  title = {Optimal {{Inapproximability Results}} for {{MAX}}-{{CUT}} and {{Other}} 2-{{Variable CSPs}}?},
  author = {Khot, Subhash and Kindler, Guy and Mossel, Elchanan and O'Donnell, Ryan},
  year = 2007,
  month = jan,
  journal = {SIAM Journal on Computing},
  volume = {37},
  number = {1},
  pages = {319--357},
  issn = {0097-5397, 1095-7111},
  doi = {10.1137/S0097539705447372},
  urldate = {2024-07-26},
  langid = {english}
}

@article{KMS23,
  title = {Pseudorandom Sets in {{Grassmann}} Graph Have Near-Perfect Expansion},
  author = {Khot, Subhash and Minzer, Dor and Safra, Muli},
  year = 2023,
  month = jul,
  journal = {Annals of Mathematics},
  volume = {198},
  number = {1},
  issn = {0003-486X},
  doi = {10.4007/annals.2023.198.1.1},
  urldate = {2024-09-10}
}

@misc{KS20,
  title = {Hardness of {{Approximate Nearest Neighbor Search}} under {{L-infinity}}},
  author = {Ko, Young Kun and Song, Min Jae},
  year = 2020,
  month = nov,
  number = {arXiv:2011.06135},
  eprint = {2011.06135},
  primaryclass = {cs},
  publisher = {arXiv},
  doi = {10.48550/arXiv.2011.06135},
  urldate = {2024-09-25},
  archiveprefix = {arXiv}
}

@article{MR08,
  title = {Two-Query {{PCP}} with Subconstant Error},
  author = {Moshkovitz, Dana and Raz, Ran},
  year = 2008,
  month = jun,
  journal = {J. ACM},
  volume = {57},
  number = {5},
  pages = {29:1--29:29},
  issn = {0004-5411},
  doi = {10.1145/1754399.1754402},
  urldate = {2025-09-29}
}

@misc{MT18,
  title = {Mildly {{Exponential Time Approximation Algorithms}} for {{Vertex Cover}}, {{Uniform Sparsest Cut}} and {{Related Problems}}},
  author = {Manurangsi, Pasin and Trevisan, Luca},
  year = 2018,
  month = jul,
  number = {arXiv:1807.09898},
  eprint = {1807.09898},
  primaryclass = {cs},
  publisher = {arXiv},
  doi = {10.48550/arXiv.1807.09898},
  urldate = {2026-02-01},
  archiveprefix = {arXiv}
}

@misc{MW18,
  title = {Oracle {{Separations}} for {{Quantum Statistical Zero-Knowledge}}},
  author = {Menda, Sanketh and Watrous, John},
  year = 2018,
  month = jan,
  number = {arXiv:1801.08967},
  eprint = {1801.08967},
  primaryclass = {quant-ph},
  publisher = {arXiv},
  doi = {10.48550/arXiv.1801.08967},
  urldate = {2024-09-24},
  archiveprefix = {arXiv}
}

@book{Pei16,
  title = {A {{Decade}} of {{Lattice Cryptography}}},
  author = {Peikert, Chris},
  year = 2016,
  publisher = {now},
  doi = {10.1561/0400000074},
  langid = {english}
}

@inproceedings{Rao08,
  title = {Parallel Repetition in Projection Games and a Concentration Bound},
  booktitle = {Proceedings of the Fortieth Annual {{ACM}} Symposium on {{Theory}} of Computing},
  author = {Rao, Anup},
  year = 2008,
  month = may,
  pages = {1--10},
  publisher = {ACM},
  address = {Victoria British Columbia Canada},
  doi = {10.1145/1374376.1374378},
  urldate = {2024-09-03},
  isbn = {978-1-60558-047-0},
  langid = {english}
}

@article{Raz98,
  title = {A {{Parallel Repetition Theorem}}},
  author = {Raz, Ran},
  year = 1998,
  month = jun,
  journal = {sicomp},
  volume = {27},
  number = {3},
  pages = {763--803}
}

@article{Reg09,
  title = {On Lattices, Learning with Errors, Random Linear Codes, and Cryptography},
  author = {Regev, Oded},
  year = 2009,
  month = sep,
  journal = {Journal of the ACM},
  volume = {56},
  number = {6},
  pages = {34:1--34:40},
  issn = {0004-5411},
  doi = {10.1145/1568318.1568324},
  urldate = {2020-08-31}
}

@article{Sch87,
  title = {A Hierarchy of Polynomial Time Lattice Basis Reduction Algorithms},
  author = {Schnorr, C. P.},
  year = 1987,
  month = jan,
  journal = {Theoretical Computer Science},
  volume = {53},
  number = {2},
  pages = {201--224},
  issn = {0304-3975},
  doi = {10.1016/0304-3975(87)90064-8},
  urldate = {2024-11-01}
}

@article{SH90,
  title = {Power Indices and Easier Hard Problems},
  author = {Stearns, R. E. and Hunt, H. B.},
  year = 1990,
  month = dec,
  journal = {Mathematical systems theory},
  volume = {23},
  number = {1},
  pages = {209--225},
  issn = {1433-0490},
  doi = {10.1007/BF02090776},
  urldate = {2024-09-23},
  langid = {english}
}

@phdthesis{Ste10,
  title = {On the {{Complexity}} of {{Unique Games}} and {{Graph Expansion}}},
  author = {Steurer, David},
  year = 2010,
  urldate = {2024-07-23},
  school = {Princeton University}
}

@misc{Ste24,
  title = {Security Analysis of the {{iMessage PQ3}} Protocol},
  author = {Stebila, Douglas},
  year = 2024,
  number = {2024/357},
  urldate = {2024-09-26}
}

@article{TSSW00,
  title = {Gadgets, {{Approximation}}, and {{Linear Programming}}},
  author = {Trevisan, Luca and Sorkin, Gregory B. and Sudan, Madhu and Williamson, David P.},
  year = 2000,
  month = jan,
  journal = {SIAM Journal on Computing},
  volume = {29},
  number = {6},
  pages = {2074--2097},
  issn = {0097-5397, 1095-7111},
  doi = {10.1137/S0097539797328847},
  urldate = {2024-08-30},
  langid = {english}
}

@inproceedings{Vas15,
  title = {Hardness of Easy Problems: {{Basing}} Hardness on Popular Conjectures Such as the Strong Exponential Time Hypothesis},
  booktitle = {10th International Symposium on Parameterized and Exact Computation ({{IPEC}} 2015)},
  author = {Vassilevska Williams, Virginia},
  editor = {Husfeldt, Thore and Kanj, Iyad},
  year = 2015,
  series = {Leibniz International Proceedings in Informatics (Lipics)},
  volume = {43},
  pages = {17--29},
  publisher = {Schloss Dagstuhl -- Leibniz-Zentrum f\"ur Informatik},
  address = {Dagstuhl, Germany},
  issn = {1868-8969},
  doi = {10.4230/LIPIcs.IPEC.2015.17},
  isbn = {978-3-939897-92-7},
  urn = {urn:nbn:de:0030-drops-55683}
}

@article{Wil05,
  title = {A New Algorithm for Optimal 2-Constraint Satisfaction and Its Implications},
  author = {Williams, Ryan},
  year = 2005,
  month = dec,
  journal = {Theoretical Computer Science},
  series = {Automata, {{Languages}} and {{Programming}}: {{Algorithms}} and {{Complexity}} ({{ICALP-A}} 2004)},
  volume = {348},
  number = {2},
  pages = {357--365},
  issn = {0304-3975},
  doi = {10.1016/j.tcs.2005.09.023},
  urldate = {2025-01-03}
}

\newpage

\appendix
\section{Other restrictions on natural fine-grained reductions to CVP}
\label{sec:other-difficulties}

\subsection{Further barriers against hardness of approximate lattice problems} 

\label{sec:binary-cvp-upper-bound-contradicts-seth}

As described in \cref{sec:main-reduction}, it is quite natural to assume binary coefficients when creating reductions to lattice problems; all the fine-grained reduction gadgets used by \cite{BGS17, ABGS21} and this work do so. Lower bounds on \gammacvp shown using such reductions are then also bounds on \gammabinarycvp.

One of the main open problems raised in \cite{BGS17, ABGS21} is whether or not we can prove SETH-based fine-grained lower bounds for \gammacvp for larger approximation ratios $\gamma$ than those in \cite{BGS17, ABGS21}. We cannot hope for such results for some larger constant $\gamma$ for \gammacvptwo due to \cite{AK23} and for \gammacvp due to \cite{EV22}; this still leaves open the possibility of improving $\gamma$ beyond 2 or 3. We remark that even such small improvements are impossible using reductions similar to theirs, even for odd $p$, because those reductions work for at least one of \gammabinarycvp[1] and \gammabinarycvp[2]. Since there are $2^{ ( 1 - \Omega(1) ) n }$-time algorithms for those problems with $\gamma$ as small as 2 as shown in \cref{rmk:alg-for-binary-cvp-2}, any reduction from \ksat to those problems would refute SETH.

\subsection{Natural reductions from \threesat to \maxtwolintwo require $4 n / 3$ vertices}

A natural reduction is defined as belonging to the following set of reductions in \cite{ABGS21}. 

\begin{definition}[{\cite[Definition 6.5]{ABGS21}}]
  A \emph{natural reduction} from $k$-SAT to $\mathrm{CVP}_2$ is a mapping from $k$-SAT instances on $n$ variables to $\mathrm{CVP}_2$ instances $B \in \mathbb{R}^{d \times n'}, \bm{t} \in \mathbb{R}^d$ such that there exists a fixed function $f : \{ 0, 1 \}^n \rightarrow \mathbb{Z}^{n'}$ with the following property: if the input $k$-SAT instance is satisfiable, then $\bm{x} \in \{ 0, 1 \}^n$ is a satisfying assignment if and only if $f (\bm{x}) \in \mathrm{CVP}_2 (\bm{t}, B)$, where $\mathrm{CVP}_2 (\bm{t}, B) \coloneqq \{ \bm{z} \in \mathbb{Z}^n : \| B \bm{z} - \bm{t} \|_2 = \mathrm{dist}_2 (\bm{t}, \mathcal{L}) \}$ is the set of the coordinates of the closest lattice points to $\bm{t}$. The mapping and the function may not be efficiently computable. 
\end{definition}

The intuition behind the definition is that a natural reduction maps any $k$-SAT variable assignment to a point in the lattice. If the $k$-SAT instance is satisfiable, then the corresponding point in the lattice to the satisfying assignment is close to the target vector. If the instance is not satisfiable, any of the corresponding points in the lattice to the assignments are far from the target vector. Then they showed the following theorem.

\begin{theorem}[{\cite[Theorem 6.6 ]{ABGS21}}]
  Every natural reduction from 3-SAT on $n$ variables to $\mathrm{CVP}_2$ on rank $n'$ lattices must have $n' > 4 (n - 2) / 3$. 
\end{theorem}

Combining this theorem with \cref{thm:weighted-main-theorem} implies the following theorem, as the rank of our reduction is exactly equal to the number of vertices in the underlying Max-Cut instance and our reduction is also a natural reduction.

\begin{theorem}
  \label{thm:natural-reductions-ksat-mc-four-thirds}
  Any reduction from 3-SAT on $n$ variables to \minusesgap \maxtwolintwo on $n'$ variables such that there is a fixed function $f : \{0,1\}^n \rightarrow \{0,1\}^{n'}$ such that $\bm{x} \in \{0,1\}^n$ is a satisfying assignment if and only if $f(\bm{x})$ satisfies at least $1-\varepsilon$ constraints, must have $n' > 4 (n - 2) / 3$.
\end{theorem}

\end{document}